\documentclass[11pt]{article}
\usepackage{sn-preamble} 
\usepackage{tikz}
\usepackage{verbatim}
\usepackage{cancel}
\usepackage{caption}
\usepackage[subpreambles]{standalone}

\usepackage{sectsty}
\partfont{\Large}

\usetikzlibrary{decorations.pathreplacing,angles,quotes}
\usetikzlibrary{shapes.geometric}
\usetikzlibrary{patterns}

\usepackage{mathtools} 
\usepackage{bm} 

\newcommand{\ignore}[1]{}

\newcommand{\ham}{\mathrm{ham}}
\newcommand{\Bin}{\mathrm{Bin}}
\newcommand{\col}{\mathrm{col}}

\newcommand{\reldist}{\mathsf{rel}\text{-}\mathsf{dist}}

\newcommand{\dtv}{\mathrm{d}_{\mathrm{TV}}}

\newcommand{\Dyes}{\calD_{\mathrm{yes}}}
\newcommand{\Dno}{\calD_{\mathrm{no}}}

\newcommand{\diff}{\mathrm{diff}}

\newcommand{\sphere}{\mathrm{Sphere}}

\newcommand{\MQ}{\mathrm{MQ}}
\newcommand{\SAMP}{\mathrm{SAMP}}

\begin{document}

\title{
Halfspaces are hard to test with relative error
}

\author{
\begin{tabular}{ccc}
    Xi Chen & Anindya De & Yizhi Huang \\
    Columbia University & University of Pennsylvania & Columbia University \\
    \url{xichen@cs.columbia.edu} & \url{anindyad@cis.upenn.edu} & \url{yizhi@cs.columbia.edu} \\
    \\
    Shivam Nadimpalli & Rocco A. Servedio & Tianqi Yang \\
    MIT & Columbia University & Columbia University \\
    \url{shivamn@mit.edu} & \url{rocco@cs.columbia.edu} & \url{tianqi@cs.columbia.edu}
\end{tabular}
}

\date{\today}

\pagenumbering{gobble}

\maketitle  

\begin{abstract}
Several recent works \cite{CDHLNSY25,CPPS25conjunctionDL,CPPS25junta,rel-error-DNF} have studied a model of property testing of Boolean functions under a \emph{relative-error} criterion.  In this model, the distance from a target function $f: \{0,1\}^n \to  \{0,1\}$ that is being tested to a function $g$ is defined relative to the number of inputs $x$ for which $f(x)=1$; moreover, testing algorithms in this model have access both to a black-box oracle for $f$ and to independent uniform satisfying assignments of $f$.  The motivation for this model is that it provides a natural framework for testing \emph{sparse} Boolean functions that have few satisfying assignments, analogous to well-studied models for property testing of sparse graphs.

The main result of this paper is a lower bound for testing \emph{halfspaces} (i.e.,~linear threshold functions) in the relative error model:  we show that $\tilde{\Omega}(\log n)$ oracle calls are required for any relative-error halfspace testing algorithm over the Boolean hypercube $\{0,1\}^n$.  This stands in sharp contrast both with the constant-query testability (independent of $n$) of halfspaces in the standard model \cite{MORS10}, and with the positive results for relative-error testing of many other classes given in \cite{CDHLNSY25,CPPS25conjunctionDL,CPPS25junta,rel-error-DNF}. Our lower bound for halfspaces gives the first example of a well-studied class of functions for which relative-error testing is provably more difficult than standard-model testing. 
\end{abstract}

\newpage

\setcounter{page}{1}
\pagenumbering{arabic}

\newpage

\section{Introduction} \label{sec:intro}

Since the early seminal works of \cite{BLR93,GGR98}, the field of \emph{property testing of Boolean functions $f: \zo^n \to \zo$} has developed into a rich and fruitful area of inquiry \cite{Goldreich17book}.  As the field has matured, researchers have gone beyond the original model (of testing using black-box oracle queries, with distance between functions measured using the uniform distribution over $\zo^n$) in a number of ways, by considering various refinements and extensions of the original model.  These variants include, among others, \emph{distribution-free} testing (first introduced by \cite{HalevyKushilevitz:07}, see \cite{CLSSX18, ChenPatel22, chen2024distribution, FlemingY20} for recent representative work); \emph{tolerant} testing (first introduced by \cite{parnas2006tolerant}, see \cite{DMN19, de2021robust, BCELR19, berman2022tolerant} for recent representative work); \emph{active} testing (first introduced by \cite{BBBY12}, see \cite{backurs2020active, SODACC} for recent representative work); and \emph{sample-based} testing (first introduced by \cite{GGR98}, see \cite{CDS20stoc, blais2021vc} for recent representative work).

The recent work of \cite{CDHLNSY25}  introduced a new model of Boolean function property testing, which is called \emph{relative-error} property testing.
To motivate this new model, recall that in the standard model of Boolean function property testing, a testing algorithm for a class of functions ${\cal C}$ makes  black-box oracle calls to the unknown function $f: \zo^n \to \zo$ that is being tested, and the goal is to distinguish between the two cases that (i) $f \in {\cal C}$, versus (ii) $f$ is $\eps$-far (under the uniform distribution) from every function in ${\cal C}$, i.e.
\[
\dist(f,{\cal C}) \geq \eps, \text{~~~where~}
\dist(f,{\cal C}):=\min_{g \in {\cal C}} \dist(f,g) \text{~and~}
\dist(f,g) := {\frac {|f^{-1}(1) \hspace{0.06cm} \triangle \hspace{0.06cm} g^{-1}(1)|}{2^n}}.
\]
While this standard model is elegant and natural, it is not well suited for testing \emph{sparse} functions, i.e.~functions which have few satisfying assignments, since any such function has very small uniform-distribution distance to the constant-0 function.  (This is analogous to how the standard ``adjacency-matrix-query'' model for graph property testing is not suitable for testing properties of sparse $N$-vertex graphs which have $o(N^2)$ edges.)

To remedy this, the relative-error property testing model, as defined in \cite{CDHLNSY25}, differs from the standard  model in the following ways:

\begin{flushleft}\begin{itemize}
\item 
The distance between the target function $f$ and a function $g$ is now measured using  \emph{relative distance}, which is simply a rescaled version of the uniform-distribution distance based on the sparsity of $f$:
\[
\reldist(f,g) := {\frac {|f^{-1}(1) \hspace{0.06cm} \triangle \hspace{0.06cm} g^{-1}(1)|}{|f^{-1}(1)|}}, \quad \text{i.e.}
\quad
\reldist(f,g) = \dist(f,g) \cdot {\frac {2^n}{|f^{-1}(1)|}}.
\]
This distance measure naturally captures the distance between $f$ and $g$ ``at the scale of $f$'' for all possible scales.\footnote{As noted in \cite{CDHLNSY25}, while $\reldist(f,g)$ is not symmetric, it is easy to verify that if $\reldist(f,g) = \eps \leq 1/2$ then $\reldist(g,f)$ is also $\Theta(\eps)$; so relative distance is symmetric up to constant factors in the setting we are interested in.
}

\item In the relative-error testing model, the testing algorithm can access a \emph{sample oracle} $\SAMP(f)$ (which takes no input and returns  a uniform random satisfying assignment $\bx \sim f^{-1}(1)$), as well as the usual black-box oracle $\MQ(f)$.\footnote{Observe that without a $\SAMP(f)$ oracle, it might be impossible to find any satisfying assignments at all of a sparse function $f$ without making a huge number of black-box queries to $f$.}

\end{itemize}
\end{flushleft}

The relative-error testing model makes it possible to meaningfully test \emph{sparse} Boolean functions for properties of interest; as we discuss below, a number of recent works 
\cite{CDHLNSY25,CPPS25conjunctionDL,CPPS25junta,rel-error-DNF} have studied the relative-error testability of natural Boolean function properties, giving various \emph{upper} bounds (efficient testing algorithms) in this model.  The present paper continues this line of investigation, but establishes a \emph{lower} bound, and moreover one which is qualitatively very different from previous results obtained in the relative-error testing model \cite{CDHLNSY25,CPPS25conjunctionDL,CPPS25junta,rel-error-DNF}. 

\medskip

\noindent {\bf Standard-model testing versus relative-error testing.}  It was shown in  \cite{CDHLNSY25} that standard-model testing is (essentially) never more difficult than relative-error testing. More precisely, \cite{CDHLNSY25} showed that for any class ${\cal C}$, if ${\cal C}$ is $\eps$-relative-error testable using $q$ oracle calls, then ${\cal C}$ is $\eps$-standard-model testable using $O(q/\eps)$ oracle calls.  (Roughly speaking, this is because if $|f^{-1}(1)| \gtrsim \eps 2^n$ then a call to the $\SAMP$ oracle can be simulated by drawing $O(1/\eps)$ many uniform random examples, while if $|f^{-1}(1)| \lesssim \eps 2^n$ then $f$ is $\eps$-close to the constant-0 function so it is easy to test.)

In the other direction, a simple example, given in Appendix A of \cite{CDHLNSY25}, gives a class ${\cal C}$ of $n$-variable Boolean functions which is trivially testable to any constant error in the standard model using zero queries, but which requires $2^{\Omega(n)}$ oracle calls in the relative-error model.  This property is rather contrived and unnatural, though;\footnote{The class ${\cal C}$ consists of all functions $g: \zo^n \to \zo$ for which the number of satisfying assignments $|g^{-1}(1)|$ is an integer multiple of $2^{2n/3}$; we refer the reader to \cite{CDHLNSY25} for the simple argument establishing the bounds claimed above.} in contrast with this artificial example, a growing body of recent results suggest that relative-error testing might \emph{not} be much more difficult than standard-model testing for properties corresponding to commonly studied classes of functions ${\cal C}$.
In particular, for many well-studied classes of functions such as monotone Boolean functions; 
conjunctions; decision lists; $k$-juntas; size-$s$ decision trees; $s$-term DNF formulas; and more, the number of oracle calls required for relative-error testing is  at most some fixed polynomial in the number of oracle calls required for testing in the standard model (see \Cref{table:results}).\footnote{In fact, for some classes such as monotone and unate functions, if the function $f$ being tested is sparse, then substantially \emph{fewer} oracle calls can suffice for relative-error testing as opposed to standard-model testing; see Table~1.}   This motivates the following question:

\begin{quote}
Do all ``natural'' Boolean function properties have a polynomial relation between the\\ number of oracle calls required for relative-error testing versus standard testing? 
\end{quote}

In this context, the current paper considers the relative-error testability of the class of \emph{halfspaces} (also known as linear threshold functions, or LTFs), which is one of the most well-studied classes of functions in property testing and learning theory.  In particular, while halfspaces are known to be testable in the standard model with a \emph{constant} number of queries \cite{MORS10}, we show a \emph{superconstant} (in fact, essentially logarithmic in dimension) lower bound in the relative error model. Thus, the class of halfspaces provides a strong negative answer to the question above. 
In light of the evidence to the contrary that is provided by \Cref{table:results}, we feel that this is a potentially surprising result.

\begin{table}  
\centering
\renewcommand{\arraystretch}{1.34}
  \begin{tabular}{ @{}p{0.3\textwidth}p{0.3\textwidth}p{0.3\textwidth}@{} }
 \toprule
Class of functions & Standard-model & Relative-error \\[-0.5em]
& testing complexity  & testing complexity\\
\midrule 
Monotone  functions   & $\tilde{O}(\sqrt{n}/\eps^2)$ \cite{KMS18},& $O((\log N)/\eps)$ \cite{CDHLNSY25}\\
&   $\tilde{\Omega}(n^{1/3})$ \cite{CWX17stoc} & $\tilde{\Omega}((\log N)^{2/3})$ \cite{CDHLNSY25}\\
 \hline
Conjunctions & $\Theta(1/\eps)$ \cite{PRS02} & $\Theta(1/\eps)$ \cite{CPPS25conjunctionDL}   \\
 \hline
Decision lists  & $\tilde{\Theta}(1/\eps)$ \cite{Bshouty20} & $\tilde{\Theta}(1/\eps)$ \cite{CPPS25conjunctionDL} \\
 \hline
$k$-juntas  & $\tilde{O}(k/\eps)$ \cite{Blaisstoc09}, & $\tilde{O}(k/\eps)$ \cite{CPPS25junta} \\
 &  $\tilde{\Omega}(k)$ \cite{ChocklerGutfreund:04,Saglam18} &  $\tilde{\Omega}(k)$ \cite{ChocklerGutfreund:04,Saglam18}\\
\hline
Size-$s$ decision trees  & $\tilde{O}(s/\eps)$ \cite{Bshouty20}, & $\tilde{O}(s/\eps)$ \cite{CPPS25junta} \\
  & $\Omega(\log s)$ \cite{CGM11} & $\Omega(\log s)$ \cite{CGM11} \\
 \hline
$s$-term DNF formulas  & $\tilde{\Theta}(s/\eps)$ \cite{Bshouty20} & $(s/\eps)^{O(1)}$ \cite{rel-error-DNF} \\
\hline
\rowcolor{yellow!30!white} {\bf Halfspaces}  & $\poly(1/\eps)$ \cite{MORS10}& $\tilde{\Omega}(\log n)$ {\bf [this work]}\\
& & $(n/\eps)^{O(1)}$  \cite{DDS15}\footnotemark\\
\bottomrule
\end{tabular}
\caption{Bounds on the number of oracle calls needed for relative-error testing and standard-model testing of various classes of $n$-variable Boolean functions, to accuracy $\eps$.  For relative-error testing, the entry indicates the total number of oracle calls to either the $\SAMP$ or the $\MQ$ oracle. The parameter $N$ is the number of satisfying assignments $|f^{-1}(1)|$ of the unknown function being tested; note that $N \leq 2^n$ always. 
}
\label{table:results}
\end{table}



We now give a detailed description of our contributions in this paper and of the context for them.

\subsection{Our contribution:  Relative-error testing of halfspaces}  

\noindent {\bf Background.}
Recall that a \emph{halfspace} is a function $f: \zo^n \to \bits, f(x) = \mathsf{sign}(w \cdot x - \theta)$. Halfspaces have been studied intensively in property testing, see e.g.~\cite{MORS10, MORS:09random, de2019your, de2021robust, Harms19, ChenPatel22}. 
This study is closely connected to other topics in theoretical computer science, such as: (a) \emph{learning} halfspaces, which has been studied since the introduction of the Perceptron algorithm in the 1960s 
\cite{Block:62,Novikoff:62}
through to the present day \cite{KKMS:08, Daniely16, diakonikolas2020complexity, diakonikolas2024efficient,chandrasekaran2024smoothed} as one of the most central problems in computational learning theory; and (b) \emph{structural properties} of halfspaces, which have been studied in many works on Boolean function analysis and probability theory such as \cite{MosselOdonnell:03, MosselNeeman15, MOO:10, Borell:85, DMN13, DFKO06}.

\footnotetext{See \Cref{sec:discussion} for an explanation of a trivial $\poly(n/\eps)$ upper bound from known relative-error learning results \cite{DDS15}.}


The principal result on testing halfspaces in the standard Boolean function property testing model, due to \cite{MORS10}, is that halfspaces are $\eps$-testable using only $\poly(1/\epsilon)$ queries --- i.e.,~they are ``constant-query testable'' independent of the ambient dimension $n$.   
Recall that the algorithmic results obtained to date for relative-error testing of many other concept classes
 --- see e.g.~the results for monotone functions, conjunctions, decision lists, $k$-juntas, size-$s$ decision trees, and $s$-term DNF formulas, due to \cite{CDHLNSY25,CPPS25conjunctionDL,CPPS25junta,rel-error-DNF}, that are given in \Cref{table:results} --- 
closely align with the standard-model testing results: for each of those classes, the number of oracle calls required for relative-error testing is  at most some fixed polynomial in the number of oracle calls required for testing in the standard model.
Given this, and given the $\poly(1/\eps)$-query testability of halfspaces in the standard model, it is natural to conjecture that 
halfspaces would similarly prove to be \emph{relative-error} $\eps$-testable with $\poly(1/\eps),$ or at worst some $O_{\eps}(1),$ complexity. 
However, as our main result, described below, we show that this is not the case.

\medskip
\noindent{\bf Our main result:  Halfspaces are hard to test in the relative-error model.} 
Our main result is a super-constant lower bound for relative-error halfspace testing over $\{0,1\}^n$:

\begin{theorem} [Boolean LTF lower bound] \label{thm:Boolean-lower}
For some constant $\eps_0>0$, any relative-error $\eps_0$-testing algorithm for  halfspaces over $\{0,1\}^n$ must either draw at least $\frac{0.05\log n}{\log \log n}$ samples from $\SAMP(f)$ or make at least $n^{0.01}$ non-adaptive queries to $\MQ(f)$.
\end{theorem}

Using the standard fact that any algorithm making $q$ adaptive queries to $\MQ(f)$ can be simulated by an algorithm making $2^q$ non-adaptive queries, \Cref{thm:Boolean-lower} implies a lower bound against adaptive testers that draw fewer than $\tilde{\Omega}(\log n)$ samples and make fewer than $0.01 \log n$  queries.

\Cref{thm:Boolean-lower} reveals a dramatic difference between halfspaces and all of the function classes listed in \Cref{table:results}:  for each of those classes of functions, the number of oracle calls required for relative-error testing is  at most some fixed polynomial in the number of oracle calls used by the best known standard-model testers.  In contrast, the $\tilde{\Omega}(\log n)$ lower bound of \Cref{thm:Boolean-lower} together with the $\poly(1/\eps)$-query standard-model testing algorithm of \cite{MORS10} shows that there is an arbitrarily large gap between the complexities of standard-model versus relative-error testing for halfspaces.

Very roughly speaking, \Cref{thm:Boolean-lower} is established by constructing a suitable pair of distributions over yes-functions (halfspaces) and no-functions (functions that are relative-error far from every halfspace).  
In our constructions, the yes-functions are Hamming balls centered at random and~unknown uniform points $\bz \sim \{0,1\}^n$; the no-functions are more complicated and we do not describe them here, but they can also be thought of as being centered at random and unknown uniform points $\bz \sim \{0,1\}^n$.
At a very high intuitive level, the idea of our lower bound is that
$\tilde{\Omega}(\log n)$ random satisfying assignments are not enough to distinguish yes- functions from no- functions, and moreover are not enough to completely determine the center $\bz$. Hence after the samples have been received, there is still some uncertainty about the center $\bz$ in both the yes- case and the no- case; this uncertainty about $\bz$ is leveraged to show that even making $n^{0.01}$ many non-adaptive queries, after receiving the initial samples, does not suffice to distinguish yes-functions from no-functions. 

\medskip
\noindent {\bf Organization.}  \Cref{sec:lower-bound-techniques} gives a more detailed overview of the ideas in our lower bound.
\Cref{sec:preliminaries} covers a few simple preliminaries.  \Cref{sec:boolean-lower} gives the proof of \Cref{thm:Boolean-lower}:  \Cref{sec:yes} presents our yes- and no- distributions, \Cref{sec:sample-based} shows that testers which only draw $\tilde{\Omega}(\log n)$ samples (and no queries) cannot succeed, and \Cref{sec:queries-dont-help} shows that after $\tilde{\Omega}(\log n)$ many samples, even making $n^{0.01}$ non-adaptive queries does not enable a tester to succeed. Finally, \Cref{sec:discussion} discusses our lower bound construction and proposes some directions for future work.

\section{Technical overview of \Cref{thm:Boolean-lower}}
\label{sec:lower-bound-techniques} 

To motivate our construction, we make the easy observation (which is explicitly stated in \cite{CDHLNSY25}) that if the function $f$ that is being tested has $|f^{-1}(1)|=p2^n$, then relative-error $\eps$-testing of $f$ can be performed simply by doing standard-model $(p \eps)$-testing of $f$.
 Since LTFs over $\{0,1\}^n$ are $\eps$-testable with $\poly(1/\eps)$ queries in the standard uniform-distribution testing model \cite{MORS10}, it follows that any $\omega_n(1)$ lower bound for testing LTFs must use functions for which $p:=|f^{-1}(1)|/2^n$ is at most some $o_n(1)$ quantity. It is natural to think of such LTFs as being similar to Hamming balls with a small radius, and indeed this intuition is the starting point of our construction.

\subsection{High-level overview of the construction}

We consider a distribution $\Dyes$ of yes- functions which are halfspaces with the following structure.  For a uniform random ``center'' $\bz \in \{0,1\}^n$ and a carefully chosen fixed value of the ``radius''~$r=$ $n/\log^2 n=o(n)$, a yes- function outputs 1 on all inputs whose Hamming distance to $\bz$ is at most $r$ and 0 on all inputs at Hamming distance strictly more than $r$ from $\bz$.  (So in other words,~our yes-function LTFs are simply Hamming balls of a certain fixed radius but with an unknown center~$\bz$.)  One useful property of these functions, which is an easy consequence of our choice of $r$ (satisfying $rs=o(n)$), is that a $1-1/\omega(s)$ fraction of satisfying assignments are at Hamming distance exactly $r$ from the ``center'' $\bz$;
we denote this set of points by $\sphere_r(\bz)$.

The distribution of no- functions we consider is similar to the yes- functions in that there is~a uniform random center $\bz\in \{0,1\}^n$, and again points at Hamming distance strictly less than $r$ from $\bz$ are labeled 1 while points at Hamming distance strictly greater than $r$ from $\bz$ are labeled 0.  The difference is in the points in $\sphere_r(\bz)$.  For no- functions, points in $\sphere_r(\bz)$ are labeled, roughly speaking, by partitioning them using a collection of $O(\log s)$ randomly chosen LTFs
  into $\text{poly}(s)$ many pieces and assigning the same (random) $\{0,1\}$-label to all points in a given piece of the partition.  \Cref{sec:yes} describes the construction in detail and shows that such functions~are, with probability $\Omega(1)$, $\Omega(1)$-far in relative-error from all LTFs (\Cref{claim:no-LTF}).  The argument employs a greedy procedure to find $\Omega(1) \cdot |f^{-1}(1)|$ many disjoint ``violating 4-tuples''; given the existence of this many violating 4-tuples, an easy argument gives the claimed relative-distance bound.  

\subsection{High-level sketch of the lower bound argument} \label{sec:lower-bound-sketch}


Given the above-described yes- and no- functions, we proceed to a sketch of the lower bound argument, which goes in two stages. In the first stage we show (in \Cref{sec:sample-based}) that for any ``sample-based'' algorithm that is only given $s=\Theta({\frac {\log n}{\log \log n}})$ uniform random satisfying assignments $\bu^1,\ldots,\bu^s$ of~$f$ and cannot make any queries, a random yes- function is indistinguishable from a random no- function.  The high-level intuition  is that with high probability, 
\begin{flushleft}\begin{enumerate}
    \item[(i)] All of the received satisfying assignments (in either case) will belong to $\sphere_r(\bz)$; 
    \item[(ii)] In the yes- case, these points $\bu^1,\ldots,\bu^s$ will be indistinguishable from a uniform random sample of points that belong to $\sphere_r(\bz)$; and 
    \item[(iii)] Likewise in the no- case, these points will be indistinguishable from a uniform random sample of points that belong to $\sphere_r(\bz)$.  
\end{enumerate}\end{flushleft}
Item (i) is a simple consequence of the fact that a $1-1/\omega(s)$ fraction of satisfying assignments, in either case, belong to $\sphere_r(\bz)$.
Item (ii) is immediate, since if all of the satisfying assignments that are received belong to $\sphere_r(\bz)$, then their distribution is uniform random over $\sphere_r(\bz)$.
Item (iii), on the other hand, requires a more careful argument.
We consider the slightly different process of (1) drawing $O(s)$ samples from $\sphere_r(\bz)$ first (regardless of how the function labels points in $\sphere_r(\bz)$); and then (2) drawing 
  the random partition to label points in $\sphere_r(\bz)$, in particular those samples
  received in (1); and (3) finally returning the
  first $s$ samples that are satisfying assignments.
This can be shown to be close to the actual distribution of $\bu^1,\ldots,\bu^s$
  as well as the uniform distribution of $s$ samples
  from $\sphere_r(\bz)$; 
  the latter follows from the fact that it is very unlikely for any two of the $O(s)$ samples from $\sphere_r(\bz)$ drawn at the beginning to land in the same piece of the partition of $\sphere_r(\bz)$ (see \Cref{claim:samp-distribution-similar} for details).

The second and main stage of the argument (in \Cref{sec:queries-dont-help}) shows that any deterministic algorithm which can make $q$ non-adaptive queries after receiving its $s$ samples from $f^{-1}(1)$ will with high  probability receive the same $q$-bit string as a response to its queries in both the yes- and no- cases.  Given the results of the previous paragraph, for this it suffices to show the following (see \Cref{lemma:maintechlower}): Any deterministic algorithm which takes as input $s$ strings $\bu^1,\dots,\bu^s$ drawn uniformly at random from $\sphere_r(\bz)$ with a uniform center $\bz$ and outputs a point $\by\in \{0,1\}^n$, can only have an $o_n(1)$ probability of outputting a point $\by$ which is both 
\begin{quote} \centering
(a) at least Hamming distance $t=n^{0.4}$ away from each $\bu^i$, and 
(b) in $\sphere_r(\bz)$.
\end{quote}
It is sufficient to show this because the construction of the random yes- functions and the random no- functions ensures that if the string $\by$ is either Hamming-close to any $\bu^i$ or is not in $\sphere_r(\bz)$, then with high probability the yes- function and no- function will label $\by$ the same way.  To see this, note that if $\by \notin \sphere_r(\bz)$ then it will be labeled the same way in the yes- case and the no- case (because in both cases the label is 1 if $\ham(\by,\bz)<r$ and is 0 if $\ham(\by,\bz)>r$).  On the other hand, if $\by \in \sphere_r(\bz)$ but $\by$ is Hamming-close to some $\bu^i$, then with very high probability in the no- case we will have $f(\by)=f(\bu^i)=1$ (given that most likely every $\bu^i$ is far from
  those random hyperplanes used to build the partition), and in the yes- case we also have $f(\by)=1$.

Thus, it remains only to argue that any deterministic algorithm, given as input  $\bu^1,\dots,\bu^s$ as described above, can only have an $o_n(1)$ probability of outputting a point $\by$ that satisfies both (a) and (b). 
Establishing this is the main technical challenge of the lower bound proof, and the argument is rather intricate, see \Cref{proofofmaintechlower}.  However, 
some intuition for why this is true is as follows: we can think of the task of a deterministic algorithm as being that it  must choose some $\bu^i$ and choose $t' \geq t$ coordinates in $[n]$ to flip in $\bu^i$ to obtain the desired string $\by$.  Since $\bu^i\in \sphere_r(\bz)$, the goal is for exactly $t'/2$ of the bit-flips to ``move away from'' $\bz$ and exactly $t'/2$ of the bit-flips to ``move towards'' $\bz$.  Intuitively, since each $\bu^i$ is a random point at Hamming distance exactly $r=o(n)$ from $\bz$, the best way for the algorithm to select $t'/2$ coordinates whose flips will all ``move away from'' $\bz$ is to select $t'/2$ coordinates from the set $\smash{\mathsf{Consistent} :=\{j \in  [n]: \bu^1_j = \cdots =\bu^s_j\}}$, since those are naturally the coordinates for which it is most likely that they are set the same way in $\bu^i$ as in $\bz$. But a simple calculation, using the precise values of $r$ and $s$, shows that at least $\kappa=\Omega(1/n^{0.1})$-fraction of the coordinates in $\mathsf{Consistent}$ are actually coordinates which are set \emph{the opposite way} in all of $\bu^1,\dots,\bu^s$ from how they are set in $\bz$.  Thus, the intuitively optimal strategy of flipping $t'/2$ coordinates from $\mathsf{Consistent}$ will only have a $\approx (1-\kappa)^{t'/2}$ probability of moving exactly $t'/2$ steps away from $\bz$. Instead, it will move roughly $t'/2 - 2 \cdot \Bin(t'/2,\kappa)$ steps away from $\bz$. Since $t' =\Omega(n^{0.4})\gg1/\kappa$, there is ``a lot of uncertainty'' in the distance from $\bz$ that results from flipping these $t'/2$ coordinates, and there is at most an $o_n(1)$ chance that the final string $\by$ has Hamming distance exactly $r$ from $\bz$.  

We close this overview by noting that it is not clear how to formalize a rigorous argument along these lines of the above informal reasoning, since the algorithm can pick any $\by$ as a function of $\bu^1,\ldots,\bu^s$. 
So the actual proof uses a rather different, and delicate, combinatorial and probabilistic argument.


\section{Preliminaries} \label{sec:preliminaries}

We use boldfaced letters such as $\bx, \boldf,\bA$, etc. to denote random variables (which may be \mbox{real-valued,} vector-valued, function-valued, or set-valued; the intended type of the random variable will be clear from the context).
We write $\bx \sim \calD$ to indicate that the random variable $\bx$ is distributed according to probability distribution $\calD$.

We will write $(e_i)_{i=1}^n$ for the collection of standard basis vectors in $\R^n$. Given two sets $A$ and $B$, we use $A \, \triangle \, B$ to denote their symmetric difference, i.e. $A\,\triangle\,B = (A\setminus B) \cup (B\setminus A)$. 
We write $\mathbb{N}$ for the set $\mathbb{N}=\{0,1,2,\dots\}$.\medskip




\noindent {\bf Relative-error testing.}
We recall the relative-error testing model that was introduced in \cite{CDHLNSY25}.
A \emph{relative-error} testing algorithm for a class of functions ${\cal C}$ over $\zo^n$ has oracle access to a black-box oracle $\MQ(f)$ for $f$, and also has access to a $\SAMP(f)$ oracle which, when called, returns an element $\bx$ drawn uniformly at random from $f^{-1}(1)$.
A relative-error testing algorithm for ${\cal C}$ must output ``yes'' with high probability (say at least 9/10; this success probability can be easily amplified) if $f \in {\cal C}$, and must output ``no'' with high probability (again, say at least 9/10) if $\reldist(f,{\cal C}) \geq \eps$, where
$\reldist(f,{\cal C})=\min_{g \in {\cal C}}\reldist(f,g)$ and 
  the relative distance between $f$ and $g$ is defined as
\[
\reldist(f,g) = {\frac {|f^{-1}(1) \ \triangle \ g^{-1}(1)|}{|f^{-1}(1)|}}.
\]


%
%


%



\def\ham{\textsf{ham}}
\def\diff{\textsf{diff}}
\def\bz{\boldsymbol{z}}
 
\section{An $\tilde{\Omega}(\log n)$ lower bound for relative-error LTF testing: Proof of \Cref{thm:Boolean-lower}}
\label{sec:boolean-lower}

In this section, we consider non-adaptive algorithms which first receive $s=\frac{0.05\log n}{\log \log n}$ samples from $f^{-1}(1)$ and then make $q=n^{0.01}$ many non-adaptive queries, and show that such algorithms cannot test LTFs under the relative-error model. Note that while these $q$ queries are non-adaptive in the sense that they all are chosen ``in parallel,'' they may depend on the outcome of the $s$ samples.

\def\bn{\{0,1\}^n}


\noindent 

\subsection{The yes- and no- distributions} \label{sec:yes}

We 
write $\ham(x,y)$ 
   to denote the Hamming distance between 
 $x,y \in \bn$.
Given $z \in \bn$, we write $\sphere_r(z)$ to denote 
the Hamming sphere of radius $r$ around $z$:
$$\sphere_r(z):=\big\{x \in \bn:\ham(x,z)=r\big\}.$$ 
Let $0^n$ denote the all-$0$ string in $\{0,1\}^n$.
Then $\sphere_r(0^n)$ contains strings in $\{0,1\}^n$ of Hamming weight $r$,
  and we can equivalently define $\sphere_r(z)$ as 
  the set of $x\oplus z$ for $x\in \sphere_r(0^n)$.

We introduce the following parameters:
\begin{equation}\label{eq:parameters1}
\delta:=\frac{1}{\log^2 n},\qquad r:=\delta n \qquad \text{and}\qquad
s:=\frac{0.05\log n}{\log \log n}
\end{equation} so that $\delta$ and $s$ satisfy
  the following two conditions:
$$
\delta s=o_n(1)\qquad\text{and}\qquad \delta^s=\frac{1}{n^{0.1}}.
$$
For simplicity we assume below that $r$ is an even
integer.\medskip

\noindent 
{\bf The yes- distribution.} 
A random yes- function is obtained by first choosing  a random ``center'' $\bz \in \bn$, and then taking the function to be the indicator of a Hamming ball of radius $r$ centered at $\bz$.
More formally, a draw of a function $\boldf \sim \Dyes$ is performed as follows:

\begin{enumerate}

\item Draw a uniform $\bz \sim \bn$.

\item For $x \in \bn$, 
\[
\boldf(x) = 
 f_{\bz}(x)=
\begin{cases}
1 & \text{if~}\ham(x,\bz) \leq r\\
0 & \text{if~}\ham(x,\bz) > r
\end{cases}.
\]
\end{enumerate}
It is easy to see that every function
  $f$ in the support of $\Dyes$ is an LTF.\medskip


\noindent\textbf{The no- distribution.} In words, a random no- function is obtained by first choosing  a random ``center'' $\bz \sim \bn$ of a Hamming ball of radius $r$. The function outputs 1 on points that lie in the interior of the Hamming ball and 0 on points that lie outside the Hamming ball. (Note that thus far, this is the same as a random yes- function.)  For points that lie exactly on the surface of the Hamming ball (i.e.~points in $\sphere_r(\bz)$), the output is determined by a randomly chosen \emph{partition} of the radius-$r$ Hamming sphere into disjoint pieces (defined by $O(\log s)$ many random hyperplanes); each piece is assigned a random bit from $\zo$, and all points in the piece have that bit as the output value.

More precisely,
a draw of a function $\bg \sim \Dno$ is performed as follows:
%
\begin{flushleft}
\begin{enumerate}
\item Draw a uniform $\bz \sim \bn$.  For conciseness let $\bX$ denote $\sphere_r(\bz)$.

\item Let $t=10\log s$.
Let $\bzeta^1,\dots,\bzeta^t$ be $t$ independent vectors each drawn uniformly from $\{\pm 1\}^n$.
We define a partition of $\bX$ into $2^t=s^{10}$ disjoint pieces $\bar{\bX}=(\bX_{\bar{b}})_{\bar{b} \in \{0,1\}^{t}}$ as follows:
\[
\bX_{\bar{b}}=
\bX_{(b_1,\dots,b_t)} :=
\left\{x \in \bX: \Indicator\big[\bzeta^i \cdot (x\oplus \bz) \geq  0\big] = b_i \text{~for all~}i\in [t]\right\}.
\] 
Note that $x\oplus \bz\in \sphere_r(0^n)$ so 
  $\bzeta^i\cdot (x\oplus \bz)\ge 0$ iff among the $r$ many 
  indices $j\in [n]$ with $(x\oplus \bz)_j=1$, the number of $+1$'s in $\bzeta^i$ is 
  at least as large as the number of $-1$'s.

\item For each string $\bar{b} \in \zo^{t}$, independently draw a random bit $\ba_{\bar{b}} \in \zo.$ Let $\bar{\ba} \in \zo^{\zo^{t}}$ denote the $2^t$-bit string consisting  of all $2^t$ of the $\ba_{\bar{b}}$'s.

\item Finally we define $\bg$ as follows: For $x \in \zo^n$, 
\[
\bg(x) = 
g_{\bz,\bar{\bX},\bar{\ba}}(x)=
\begin{cases}
1 & \text{if~}\ham(x,\bz) < r\\
0 & \text{if~}\ham(x,\bz) > r\\
\ba_{\bar{b}} & \text{if~} x \in \sphere_r(\bz) \text{~and~} x \in \bX_{\bar{b}}
\end{cases}.
\]
\end{enumerate}
\end{flushleft}
The following claim says that $\bg\sim\Dno$ is relative-error far from LTFs with high probability.











\begin{lemma} \label{claim:no-LTF}
With probability at least $0.12$, 
a draw of $\bg \sim \Dno$ has $\reldist(\bg,\text{{LTF}})=\Omega(1).$
\end{lemma}
\begin{proof}
It suffices to consider the case when $\bz=0^n$ and we write
  $g_{\bar{\bX},\bar{\ba}}$ to denote $g_{0^n,\bar{\bX},\bar{\ba}}$
  below.

We say $(x^1, x^2, x^3, x^4) \in (\{0, 1\}^n)^4$ is a \emph{good} $4$-tuple (of strings in $\sphere_r(0^n)$) if
  there exist four pairwise disjoint subsets $J_1,J_2,J_3,J_4\subset [n]$ of size $r/2$ each such that 
\begin{itemize}
\item $x^1_j=x^2_j=1$ and $x^3_j=x^4_j=0$ for all $j\in J_1$;
\item $x^1_j=x^2_j=0$ and $x^3_j=x^4_j=1$ for all $j\in J_2$;
\item $x^1_j=x^3_j=1$ and $x^2_j=x^4_j=0$ for all $j\in J_3$;
\item $x^1_j=x^3_j=0$ and $x^2_j=x^4_j=1$ for all $j\in J_4$; and
\item $x^i_j=0$ for all $i\in \{1,2,3,4\}$ and $j\notin J_1\cup J_2\cup J_3 \cup J_4.$
\end{itemize}
We say that two good $4$-tuples are disjoint if they are disjoint as two
  sets of size $4$ each.

Moreover, given any function $h: \{0, 1\}^n \to \{0, 1\}$, we call $w \coloneqq (x^1, x^2, x^3, x^4) \in (\{0, 1\}^n)^4$ a \emph{violating} $4$-tuple for $h$ if $w$ is good and $h(x^1) = h(x^4) \ne h(x^2) = h(x^3)$, because this shows that $h$ cannot be an LTF. 
To see this, suppose $h(x) = \Indicator[\xi \cdot x \geq \theta]$ is an LTF and $(x^1, x^2, x^3, x^4)$ is a violating $4$-tuple for $h$. Then (noting that $h(x)$ is well defined over $\mathbb{R}^n$) 
\[
h\pbra{{\frac {x^1 + x^4} 2}}
=
h(x^4) \neq
h(x^2)
=
h\pbra{{\frac {x^2 + x^3} 2}},
\]
but this contradicts the fact that the two vectors $({x^1 + x^4})/ 2$ and $({x^2 + x^3})/ 2$ are identical.
%
%

\Cref{claim:no-LTF} follows from the following two claims:

\begin{claim}\label{claim:prob-no-function-makes-4-tuple-violating}
Let $w=(x^1, x^2, x^3, x^4)$ be a good $4$-tuple. Then we have
    \[\Prx_{\bar\bX, \bar\ba}\big[w\text{ is violating for }g_{\bar\bX, \bar\ba}\big] \ge 0.124.\]
\end{claim}
 \begin{claim}\label{claim:Omega-1-disjoint-good-4-tuples}
There is a collection of $\Omega\big({n\choose r}\big)$ many pairwise disjoint good $4$-tuples in $\sphere_r(0^n)$.
    \end{claim}

We delay the proof of these two claims and first use them to prove \Cref{claim:no-LTF}.
   Let $W$ denote a collection of $\Omega({n\choose r})$ many pairwise disjoint good $4$-tuples from \Cref{claim:Omega-1-disjoint-good-4-tuples}. 
Let $\bY$ be the random variable that denotes the number of 
  $4$-tuples in $W$ that are violating in $g_{\bar\bX,\bar\ba}$.
Then by \Cref{claim:prob-no-function-makes-4-tuple-violating}, 
    \[\begin{split}
        \Ex_{\bar\bX, \bar\ba}\big[\bY\big]\ge 0.124\cdot |W|.
    \end{split}\]
    Therefore, by Markov inequality, we have
    \[\Prx_{\bar\bX, \bar\ba}\big[\bY \ge 0.005|W|\big] \ge 0.12.\]
When the event above happens, 
  at least one string in each violating $4$-tuple needs to be corrected
  to make the function $g_{\bar\bX,\bar\ba}$ an LTF and thus, at least 
  $\Omega({n\choose r})$ strings need to be correct, which implies that 
  the $g_{\bar\bX,\bar\ba}$ has relative distance $\Omega(1)$ from any LTF
  using the fact that
$$
g^{-1}_{\bar\bX,\bar\ba}(1)\le \sum_{k=0}^r {n\choose r} =O\left({n\choose r}\right).
$$
    This concludes the proof of \Cref{claim:no-LTF}.
\end{proof}

Now we prove \Cref{claim:prob-no-function-makes-4-tuple-violating} and 
  \Cref{claim:Omega-1-disjoint-good-4-tuples}:
    
\begin{proofof}{\Cref{claim:prob-no-function-makes-4-tuple-violating}}
Let $(x^1,x^2,x^3,x^4)$ be a good $4$-tuple with respect to $J_1,J_2,J_3$ and $J_4$.
We show that with probability at least $1-o_n(1)$ (over the draw of $\bar{\bX}$), 
  $x^1,x^2,x^3$ and $x^4$ lie in distinct parts of the partition.
The claim follows because when they do lie in
  distinct parts, the probability of 
$$
\left(g_{\bar\bX,\bar\ba}(x^1),
g_{\bar\bX,\bar\ba}(x^2),
g_{\bar\bX,\bar\ba}(x^3),
g_{\bar\bX,\bar\ba}(x^4)\right)
=(1,0,0,1)\ \text{or}\ (0,1,1,0)
$$
is $1/8$.
So the probability that the $4$-tuple is violating is at least $(1-o_n(1))\cdot (1/8)\ge 0.124$.

To see that they lie in distinct parts of the partition $\bar\bX$ with high probability, we take any two points $x,y$ from $\{x^1,x^2,x^3,x^4\}$ and show that they lie in distinct parts with probability at least $1-o_n(1)$. The claim then follows from a union bound.

There are two situations we need to consider.
We start with the easier case when $x=x^1$ and $y=x^4$
  (with disjoint supports).
Let $\bzeta\sim \{\pm 1\}^n$. Because the supports of $x$ and $y$ are disjoint,
  the probability that $x$ and $y$ lying on different sides of $\bzeta$ is at least $1/3$.
Therefore, with $t=10\log s$ many $\bzeta$'s drawn in $\bar\bX$, the probability that 
  $x,y$ lying in the same part is at most $(2/3)^{t}=o_n(1)$.

Finally we consider the case when $x=x^1$ and $y=x^2$, where $x$ is supported on $J_1\cup J_3$ and $y$ is supported on $J_1\cup J_4$.
Let $\bzeta\sim \{\pm 1\}^n$ and  
 let $\bs_1, \bs_3, \bs_4$ denote the sums of the coordinates of $\bzeta$ over indices in $J_1, J_3$ and $J_4$, respectively. By
independence across the coordinates of $\bzeta$, the following compound event happens with probability
at least $1/5$:
$$
|\bs_1| < 0.01\sqrt{r/2}, \quad 
\bs_3 > 0.01\sqrt{r/2},\quad\text{and}\quad 
\bs_4 < -0.01\sqrt{r/2}.
$$
When this happens $x$ and $y$ must lie on different sides of $\bzeta$. The rest of the argument is similar to the previous case. 
This finishes the proof of the claim.
\end{proofof}

    \begin{proofof}{\Cref{claim:Omega-1-disjoint-good-4-tuples}}
        Consider the following bipartite graph $G = ((U, W), E)$: $U = \sphere_r(0^n)$, $W$ is the set of all good $4$-tuples in $\sphere_r(0^n)$, and $(u,w) \in E$ for $u \in U, w \in W$ if and only if $u$ is one of the strings in the $4$-tuple $w$.

        By symmetry, every vertex in $U$ shares the same degree, which we denote by $d$. On the other hand, every vertex in $V_2$ has degree $4$. So $G$ is a bi-regular bipartite graph.
        Note that if two vertices in $W$ have disjoint neighborhoods, then they are two disjoint good $4$-tuples. Therefore, it suffices to show that there are $\Omega({n\choose r})=\Omega(|U|)$ vertices in $W$  that have pairwise disjoint neighborhoods.

        Let $H = \emptyset$. We repeat the following greedy process to find pairwise disjoint good $4$-tuples: 
        \begin{flushleft}\begin{enumerate}
            \item Pick an arbitrary $w = (x^1, x^2, x^3, x^4)$ in $W$ and add $w$ into $H$.
            \item Remove $w$ from $W$, $x^1, x^2, x^3, x^4$ from $U$, and all edges incident to one of them from $E$.
        \end{enumerate}\end{flushleft}
        Note that when we add a good $4$-tuple $w$ in $W$ to $H$, we remove every good $4$-tuple not disjoint with $w$ from $W$. Thus, the good $4$-tuples in $H$ are pairwise disjoint. Moreover, since every vertex in $U$ has degree $d$ and every vertex in $W$ has degree $4$ at the beginning, each round of the process we remove at most $4d + 1$ vertices from $W$. Therefore, we have
        \[|H| \ge \frac{|W|}{4d + 1} = \frac{d|U|/4}{4d + 1} = \Omega(|U|). \]
This finishes the proof of \Cref{claim:Omega-1-disjoint-good-4-tuples}. \end{proofof}

\subsection{Sample-based testers cannot succeed}\label{sec:sample-based}
The first step to prove the lower bound is to show that an algorithm that only receives $s$ samples and does not make any queries cannot distinguish random yes- functions from random no- functions.  More precisely, in this section we will prove the following lemma:

\begin{lemma} \label{claim:samp-distribution-similar}
Fix any $z\in \{0,1\}^n$. Consider the following three distributions over $s$-tuples of points from $\bn$:
\begin{flushleft}
\begin{enumerate}

\item [(a)] 
Let $\bar{\bu}=(\bu^1,\dots,\bu^s)$ be $s$ independent uniform draws from $\sphere_r(z)$.

\item [(b)] 
Let $\bar{\bv}=(\bv^1,\dots,\bv^s)$ be $s$ independent uniform draws from $f_z^{-1}(1)$ (or equivalently, draws from the set of points with Hamming distance at most $r$ from $z$).

\item [(c)] Draw 
$\bar{\bX}$ and $\bar{\ba}$ uniformly at random. Let
  $\bg=g_{z,\bar{\bX},\bar{\ba}}$ be the no- function defined using
  $z,\bar{\bX}$ and $\bar{\ba}$, and let $\bar{\bw}=(\bw^1,\dots,\bw^s)$ be $s$ independent uniform draws from $\bg^{-1}(1)$.

\end{enumerate}
\end{flushleft}
Then we have
\begin{equation} \label{eq:uv-close}
\dtv(\bar{\bu},\bar{\bv})=\dtv\big((\bu^1,\dots,\bu^s),(\bv^1,\dots,\bv^s)\big) \leq o_n(1)
\end{equation} 
and
\begin{equation} \label{eq:vw-close}
\dtv(\bar{\bv},\bar{\bw})=\dtv\big((\bv^1,\dots,\bv^s),(\bw^1,\dots,\bw^s)\big) \leq o_n(1).
\end{equation}
\end{lemma}

\Cref{eq:vw-close} implies  that no $s$-sample sample-based tester (which makes no queries) can succeed.  \Cref{eq:uv-close} and \Cref{eq:vw-close} together imply that all three distributions are close, which we will use in \Cref{sec:queries-dont-help}  to argue about testing algorithms that can make non-adaptive queries.

\begin{proof}[Proof of \Cref{claim:samp-distribution-similar}]
By symmetry, it suffices to prove the lemma for the case when $z=0^n$.
\Cref{eq:uv-close} follows from the fact that, if $\bv$ is a uniform draw
  from points of Hamming weight at most $r$, then we have
  (recalling the fact that $\smash{{ {{n \choose k}}\big/{{n \choose k-1}}}= {\frac {n-k+1} k}}$)
\[
\Prx_{\bv}\big[\bv\in \sphere_r(0^n)\big]=
\frac{{n\choose r}}{\sum_{k=0}^{r} {n \choose k} }\leq  O(\delta).
\]
A union bound over the $s$ samples $\bv^1,\ldots,\bv^s$ 
  lead to \Cref{eq:uv-close} using $s\cdot O(\delta)=o_n(1)$. 

\def\nil{\textsf{nil}}

To prove \Cref{eq:vw-close} we consider the following coupling, where
  one case leads to a distribution that is close to $(\bv^1,\ldots,\bv^s)$
  and the other case leads to a distribution that is close to $(\bw^1,\ldots,\bw^s)$:
\begin{flushleft}\begin{enumerate}
\item First draw $\bx^1,\ldots,\bx^{3s}$ independently and uniformly from the set of all points of $\zo^n$ that have Hamming weight at most $r$.
\item In the first case, we draw $3s$ independent random bits $\bb_1,\ldots,\bb_{3s}$ 
  from $\{0,1\}$.
If the number of $1$'s is at least $s$, return $\bar{\bv}^*$ as the ordered tuple of 
  the first $s$ points in $\bx^1,\ldots,\bx^{3s}$ with $\bb_i=1$;
if the number of $1$'s is less than $s$, return $\bar{\bv}^*=\nil$.
\item In the second case, we draw $\bar{\bX}$ and $\bar{\ba}$ as before and use them
  to define $\bg$ (together with $0^n$).
If the number of $\bx^i$ with $\bg(\bx^i)=1$ is at least $s$, return
  $\bar{\bw}^*$ as the ordered tuple of the first $s$ points in $\bx^1,\ldots,\bx^{3s}$
  with $\bg(\bx^i)=1$;
if the number is less than $s$, return $\bar{\bw}^*=\nil$.
\end{enumerate}\end{flushleft}
The total variation distance between $\bar{\bv}$ and $\bar{\bv}^*$ can 
  be easily bounded from above by the probability of $\bar{\bv}^*=\nil$ which 
  is $o_n(1)$ by the choice of $s$. 
  Given that $\dtv(\bar{\bv},\bar{\bv}^*)=o_n(1)$,
\Cref{eq:vw-close} follows from the following two claims, which bound
  $\dtv(\bar{\bw},\bar{\bw}^*)$ and $\dtv(\bar{\bv}^*,\bar{\bw}^*)$, respectively.

\begin{claim} \label{claim:hehe2}
$\dtv(\bar{\bw},\bar{\bw}^*)=o_n(1).$
Moreover, with probability $1-o_n(1)$,
  $\bar{\bw}$ and $\bar{\bX}$ (as in (c)) satisfy 
\[
\bw^i\in \sphere_r(0^n)\quad\text{and}\quad
\big|\bzeta\cdot \bw^i\big|\ge n^{0.49},\qquad
\text{for all $i\in [s]$ and $\bzeta \in \bar{\bX}$.}
\]
\end{claim}
\begin{proof}
The distance $\dtv(\bar{\bw},\bar{\bw}^*)$ again can be bounded by the probability of $\bar{\bw}^*=\nil$. 
To this~end we first show that $\bx^1,\ldots,\bx^{3s}$ satisfy the following 
  two conditions with probability at least $1-o_n(1)$:
\begin{enumerate}
\item All $\bx^1,\ldots,\bx^{3s}$ lie in $\sphere_r(0^n)$; and
\item For every two points $\bx^i$ and $\bx^j$, we have 
$
\big|\{k \in [n]: \bx^i_k=\bx^j_k=1\}\big|=O(\delta^2 n)$. 
\end{enumerate}
Here the first item follows from the same argument used earlier in the proof of \Cref{eq:uv-close}. The second item follows from a {straightforward combinatorial argument to analyze a single $\bx^i,\bx^j$ pair} 
and a union bound on all ${3s \choose 2}$ pairs.

Now assuming that $x^1,\ldots,x^{3s}$ are $3s$ points that satisfy both conditions above,
  we show below that with probability at least $1-o_n(1)$, they lie in $3s$ distinct parts of
  $\bar{\bX}$.
To see that this is the case, we analyze the probability of $x^1$ and $x^2$ lying in the same 
  part of $\bar{\bX}$ and then apply a union bound across all ${3s \choose 2}$ pairs.

Let $I_1$ be the set of $i\in [n]$ with $x^1_i=1$ and $x^2_i=0$,
  $I_2$ be the set of $i\in [n]$ with $x^1_i=0$~and $x^2_i=1$,
  and $I$ be the set of $i\in [n]$ with $x^1_i=x^2_i=1$.
By the assumption on $x^1,x^2$ we have that $|I|=O(\delta^2 n)=O(\delta r)$, and $|I_1|=|I_2|=(1-O(\delta))r$.
Consider a uniform draw $\bzeta\sim \{\pm 1\}^n$ and let $\bs_1,\bs_2,\bs$ denote
  the sums of the coordinates of $\bzeta$ over indices in $I_1,I_2,I$, respectively.
By independence across the coordinates of $\bzeta$, the following compound event happens with probability at least $1/5$:
\[
|\bs|<0.01 \sqrt{r}, \quad
\bs_1 > 0.01 \sqrt{r} \quad \text{and}\quad
\bs_2 < -0.01\sqrt{r}.
\]
When this happens $x^1$ and $x^2$ must lie on different sides of 
the hyperplane defined by
$\bzeta \cdot x =  0.$
As a result, the probability of having no two points lying in the same part is at least
$$
1 - {3s \choose 2} \cdot \left({\frac 4 5}\right)^{10\log s}=1-o_n(1).
$$

Assuming that $x^1,\ldots,x^{3s}$ lie in distinct parts in $\bar{\bX}$,
  the probability of $\bar{\bw}^*=\nil$ is that of having
  less than $s$ many $1$'s in $3s$ fair coins, which is $o_n(1)$.
This finishes the proof of the first part.

For the second part of the claim, we first note that the proof above 
  shows that $(\bar{\bw},\bar{\bX},\bar{\ba})$ (as in (c))
  and $(\bar{\bw}^*,\bar{\bX},\bar{\ba})$ (as in item 3 above)  
  have total variation distance at most $o_n(1)$. So it suffices to prove the statement for the latter.
As shown above, with probability at least $1-o_n(1)$ we have
  $\bx^i\in \sphere_r(0^n)$ for all $i\in [s]$.
When this is the case, a random $\bzeta$ satisfies 
  $|\bzeta\cdot \bx^i|< n^{0.49}$ with probability at most
  $O(n^{0.49}/\sqrt{r})\le O(\log n/n^{0.01})$.
The second part then follows from a union bound over
  the $O(s\log s)$ many pairs of $\bx^i$ and $\bzeta$ in $\bar{\bX}$.
\end{proof}

\begin{claim}
$\dtv(\bar{\bv}^*,\bar{\bw}^*)=o_n(1)$.
\end{claim}
\begin{proof}
Consider the event over $\bx^1,\ldots,\bx^{3s}$ and $\bar{\bX}$ such that
\begin{enumerate}
\item All $\bx^1,\ldots,\bx^{3s}$ lie in $\sphere_r(0^n)$; and
\item All $\bx^1,\ldots,\bx^{3s}$ lie in distinct parts of $\bar{\bX}$.
\end{enumerate}
We already showed in the proof of the previous claim that this event happens with
  probability at least $1-o_n(1)$.
When this event happens, we note that $\bg(\bx^1),\ldots,\bg(\bx^{3s})$ are 
  indeed independent uniform random bits and thus, $\bar{\bv}^*$ and $\bar{\bw}^*$  share the same distribution.
\end{proof}
This finishes the proof of \Cref{claim:samp-distribution-similar}.
\end{proof}

\def\ALG{\textsf{ALG}}

\subsection{Queries do not help} \label{sec:queries-dont-help}

Let $q=n^{0.01}$.
Assume for a contradiction that
there exists an $s$-sample, $q$-query, non-adaptive algorithm that accepts $\boldf\sim \Dyes$ with probability
  at least $0.99$ (over the randomness of the draw of $\boldf$ and 
  the randomness of $s$ samples from $\boldf^{-1}(1)$)
  and rejects $\bg\sim \Dno$ with probability at least $0.99$.
Then there exists a deterministic, $s$-sample, $q$-query, non-adaptive algorithm, which we refer to as $\ALG$, that satisfies
\begin{align}\label{eq:contradiction}
\Prx_{\substack{\boldf\sim \Dyes\\ \bv^1,\ldots,\bv^s\sim\boldf^{-1}(1)}}\big[\ALG(\boldf;\bv^1,&\ldots,\bv^s)=1\big]
-\Prx_{\substack{\bg\sim \Dno\\ \bw^1,\ldots,\bw^s\sim \bg^{-1}(1)}} \big[\ALG(\bg;\bw^1,\ldots, \bw^s)=1\big]\\[-2ex] &\ \ \ \ \ \ \ \ \ \ \ \ \ \ \ \ \ \ \ \ \ \ \ \ \ \ \  \ \ \ \ \ \ \ \ \ \ \ \ \ \ \ \ \ \ \ge {0.99-(1-0.12\cdot 0.99)\ge 0.1 ,}\nonumber
\end{align}
where $\ALG(f;x^1,\ldots,x^s)=1$ means that $\ALG$ accepts when it is given
  samples $x^1,\ldots,x^s$ and query access to $f$, and $=0$ means that $\ALG$ rejects.
As $\ALG$ is deterministic and non-adaptive, equivalently
  one can view it as a tuple of $q$ maps $\Phi_i:(\{0,1\}^n)^s\rightarrow \{0,1\}^n$ {and a final combining function $\Psi: \zo^q \to \{0,1\},$}
  where $\Phi_i(x^1,\ldots,x^s)$ denotes the $i$th (non-adaptive) query that $\ALG$ makes after receiving samples $x^1,\ldots,x^s$
 {and $\Psi(b_1,\dots,b_q) \in \{0,1\}$ is the final output bit that $\ALG$ generates when its $q$ queries are answered with bits $b_1,\dots,b_q\in \{0,1\}$}.

We delay the proof of the main technical lemma, \Cref{lemma:maintechlower}, to
  \Cref{proofofmaintechlower}, and first use it to prove that \Cref{eq:contradiction} cannot
  be true in the rest of this subsection:

\begin{lemma}\label{lemma:maintechlower}
Fix any map $\Phi:(\{0,1\}^n)^s\rightarrow \{0,1\}^n$.
Draw $\bz\sim \{0,1\}^n$ uniformly at random and draw $s$ points $\bu^1,\ldots,\bu^s \sim \sphere_r(\bz)$ independently and uniformly.
Let $\by=\Phi(\bu^1,\ldots,\bu^s)$. Then the following
  event occurs with probability at least $1-o(1/q)$:
$$ 
\text{ $\by\notin \sphere_r(\bz)$, \quad or \quad
 $\ham(\by,\bu^i)<n^{0.4 }$ for some $i\in [s]$.}
$$
\end{lemma}

We use \Cref{lemma:maintechlower} to prove that \Cref{eq:contradiction} cannot hold.

First, we have from the first part of \Cref{claim:samp-distribution-similar} that 
$$
\left|\Prx_{\substack{\bz\sim \{0,1\}^n\\ \bu^1,\ldots,\bu^s\sim \sphere_r(\bz)}}\big[\ALG(f_{\bz};\bu^1,\ldots,\bu^s)=1\big]
-\Prx_{\substack{\boldf\sim \Dyes\\ \bv^1,\ldots,\bv^s\sim\boldf^{-1}(1)}}\big[\ALG(\boldf;\bv^1,\ldots,\bv^s)=1\big]
\right|=o_n(1).
$$
As a result, to show that \Cref{eq:contradiction} cannot hold,
  it suffices to show that 
$$
\left|\Prx_{\substack{\bz\sim \{0,1\}^n\\ \bu^1,\ldots,\bu^s\sim \sphere_r(\bz)}}\big[\ALG(f_{\bz};\bu^1,\ldots,\bu^s)=1\big]
-\Prx_{\substack{\bg\sim \Dno\\ \bw^1,\ldots,\bw^s\sim\bg^{-1}(1)}}\big[\ALG(\bg;\bw^1,\ldots,\bw^s)=1\big]
\right|=o_n(1).
$$
Using \Cref{claim:samp-distribution-similar} (both parts) and 
  the  coupling characterization
of total variation distance, 
there is a distribution
  $\calD$ such that 
  $$\left(\big(\bz^{1},(\bu^1,\ldots,\bu^s)\big),\big(\bz^2,(\bw^1,\ldots,\bw^s),\bar{\bX},\bar{\ba}\big)\right)\sim \calD$$ satisfies the following three conditions:
\begin{flushleft}\begin{enumerate}
\item The marginal distribution of $(\bz^1,(\bu^1,\ldots,\bu^s))$ is the 
  same as drawing $\bz^1\sim \{0,1\}^n$ uniformly and then $\bu^1,\ldots,\bu^s\sim \sphere_r(\bz^1)$ independently and uniformly.
\item The marginal distribution of $(\bz^2,(\bw^1,\ldots,\bw^s),\bar{\bX},\bar{\ba}))$
  is the same as drawing $\bz^2\sim \{0,1\}^n$, $\bar{\bX}$ and $\bar{\ba}$ uniformly
  at random as in the definition of $\Dno$, and then $\bw^1,\ldots,\bw^s\sim \bg^{-1}(1)$,
  where $\bg=g_{\bz^2,\bar{\bX},\bar{\ba}}$ is the no- function
  defined using $\bz^2,\bar{\bX}$ and $\bar{\ba}$.
\item With probability $1-o_n(1)$ over $\calD$, 
  we have $\bz^1=\bz^2$ and $(\bu^1,\ldots,\bu^s)=(\bw^1,\ldots,\bw^s)$.
\end{enumerate}\end{flushleft}

Now let $\Phi_1,\ldots,\Phi_q$ be the query maps of $\ALG$.
Consider the following event over $\calD$:
\begin{flushleft}\begin{enumerate}
\item Every $\by^i=\Phi_i(\bu^1,\ldots,\bu^s)$ satisfies  $\by^i\notin \sphere_r(\bz^1)$ or
  $\ham(\by^i,\bu^j)<n^{0.4}$ for some $j\in [s]$.
\item Every $\bw^i$  satisfies that $|\bzeta\cdot (\bz^2\oplus \bw^i)|\ge n^{0.49}$  
for every hyperplane $\bzeta$ in $\bar{\bX}$.
\end{enumerate}\end{flushleft}
We have directly from \Cref{lemma:maintechlower} (and a union bound over the $q$ functions $\Phi_1,\dots,\Phi_q$) that the first part of the event occurs
  with probability $1-o_n(1)$.
The second part of the event follows directly from 
  \Cref{claim:hehe2}.

As a result, by a union bound, a draw from $\calD$ satisfies
  the two conditions above together with 
$$
\bz^1=\bz^2\quad\text{and}\quad (\bu^1,\ldots,\bu^s)=(\bw^1,\ldots,\bw^s)
$$
with probability at least $1-o_n(1)$.
We finish the proof by showing that when this occurs,
  running $\ALG$ on $f_{\bz^1}$ with samples $\bu^1,\ldots,\bu^s$ must return
   the same answer as running $\ALG$ on $\bg$ with samples $\bw^1,\ldots,\bw^s$, 
   which contradicts \Cref{eq:contradiction}.

Let $((z^1,(u^1,\ldots,u^s),(z^2,(w^1,\ldots,w^s),\bar{X},\bar{a}))$
  be such a tuple.
Let $z=z^1=z^2$.
Given that we have  $(u^1,\ldots,u^s)=(w^1,\ldots,w^s)$,
  the queries $y^1,\ldots,y^q$ made by $\ALG$ are the same.
Then it suffices to show that $f_z(y^i)=g_{z,\bar{X},\bar{a}}(y^i)$ for all $i\in [q]$.
To see this is the case, for each query $y_i$:
\begin{flushleft}\begin{enumerate}
\item Either $y^i\notin \sphere_r(z)$, in which case trivially we have 
  the results are the same; or
\item $y^i\in \sphere_r(z)$ and has Hamming distance at most $n^{0.4 }$ from some $u^j=w^j$.
Then we have $f_z(y^i)=1$ just because $y^j\in \sphere_r(z)$.
On the other hand, $g(y^i)=1$ because 
  $y^i$ must lie in the same part  {of $\bar{X}$} as $w^j$ given  
  $\ham(y^i,w^j)\le n^{0.4}$ but $|\zeta\cdot (z\oplus w^j)|
  > n^{0.48}$ for all $\zeta$ in $\bar{X}$.
\end{enumerate}\end{flushleft}
This finishes the proof of the lower bound except for the proof of \Cref{lemma:maintechlower}.

\subsection{Proof of \Cref{lemma:maintechlower}}\label{proofofmaintechlower}

Fix a map $\Phi:(\{0,1\}^n)^s\rightarrow \{0,1\}^n$. 
Let $\bz\sim \{0,1\}^n$ and $\bu^1,\ldots,\bu^s$ be $s$ independent and uniform
  draws from $\sphere_r(\bz)$.
Given a tuple $U=(u^1,\ldots,u^s)$ in the support of $\bU=(\bu^1,\ldots,\bu^s)$,
  we write $\calZ_U$ to denote the distribution of $\bz$ conditioning on $U$.

At a high level, the proof of \Cref{lemma:maintechlower} proceeds in two steps:
\begin{flushleft}\begin{enumerate}
\item We first define a condition on $U$ capturing the notion that a tuple is ``typical'', and we show that for $(\bz,\bU)$ drawn as described above, $\bU=(\bu^1,\ldots,\bu^s)$ is typical with probability at
  least $1-1/n^{\omega_n(1)}$ (see \Cref{def:typical} and \Cref{lemma:easy1}); 
\item Given any typical $U=(u^1,\ldots,u^n)$, letting $y=\Phi(U)$ and assuming that $\ham(y,u^i)\ge n^{0.4}$ for all $i\in [s]$, 
  we show that $y\notin \sphere_r(\bz)$ with probability at least $1-o(1/q)$
  over the randomness of $\bz\sim \calZ_U$ (see \Cref{lemma:easy2}). 
\end{enumerate}\end{flushleft}
\Cref{lemma:maintechlower} then follows directly.

We start with the definition of typical tuples.
Given a tuple $U=(u^1,\ldots,u^s)$ in the support, we partition $[n]$ into $2^s$ sets:
  $\col_c(U)$ for each $c\in \{0,1\}^s$, where $i\in \col_c(U)$ iff 
  $c=(u^1_i,\ldots,u^s_i)$.
(Equivalently we view $U$ as an $s\times n$ matrix where the row vectors are $u^1,\ldots,u^s$; $\col_c(U)$ contains indices of columns that are $c$.)

To gain some intuition for the definition of typical tuples, let us consider the size of $\col_c(\bU)$~for a fixed string $c\in \{0,1\}^s$ over a random $\bU = (\bu^1, \dots, \bu^s)$. Writing $|c|$ for the Hamming weight of the $s$-bit string $c$, one may observe that for any index $i \in [n]$, if $z_i = 0$ then the probability that $i$ falls into $\col_c(\bU)$ is $\delta^{\lvert c \rvert} (1-\delta)^{s-\lvert c \rvert}$; if $z_1 = 1$, the probability  is $\delta^{s-\lvert c \rvert} (1-\delta)^{\lvert c \rvert}$. Since $\bz$ is drawn uniformly at random, in expectation, $\col_c(\bU)$ contains 
$$
\text{$\frac{n}{2}\left(\delta^{\lvert c \rvert} (1-\delta)^{s-\lvert c \rvert}\right)$ indices $i$ with $z_i = 0$ and $\frac{n}{2}\left(\delta^{s-\lvert c \rvert} (1-\delta)^{\lvert c \rvert}\right)$ indices $i$ with $z_i = 1$}.
$$ 
In the definitions below, we say that $U$ is \emph{typical} if the number of $0$'s and $1$'s
  of $\bz\sim \calZ_U$ for each $\col_c(U)$ matches the expectation above within
  a reasonable error.


\newcommand{\goodzUc}[3]{\mathsf{good}(#1,#2,#3)}
\newcommand{\goodzU}[2]{\mathsf{good}(#1,#2)}

\begin{definition}\label{def:good}
We say a triple $(z, U, c)$, where $z \in \{0,1\}^n$, $U = (u^1, \dots, u^s)$, and $c \in \{0,1\}^s$, is \emph{good}, denoted by $\goodzUc{z}{U}{c}$, if the following two conditions hold:
\begin{align*}
\left| i\in \col_c(U): z_i=0\right|&=\frac{n}{2}\left( 
   \delta^{ |c|}(1-\delta)^{s-|c|} \right)\pm n^{0.51}\quad\text{and} \\[0.8ex]
\left| i\in \col_c(U): z_i=1\right|& = \frac{n}{2}\left(\delta^{s-|c|} 
  (1-\delta)^{ |c|} \right)\pm n^{0.51}.
\end{align*}
When $(z, U, c)$ is good for all $c \in \{0,1\}^s$, we say the pair $(z,U)$ is \emph{good}, denoted by $\goodzU{z}{U}$.
\end{definition}

\begin{definition} \label{def:typical}
We say a tuple $U = (u^1, \dots, u^s)$ in the support of $\bU$ is \emph{typical} if
    \begin{equation*}
        \Prx_{\bz \in \calZ_U} \left[ \goodzU{\bz}{U} \right] \ge 1 - n^{-\omega_n(1)}.
    \end{equation*}
\end{definition}

Note that for any typical $U$, the size of each $\col_c(U)$ must be close to the expectation as well:
\begin{equation*}
\left|\col_c(U)\right|=\frac{n}{2}\left(\delta^{|c|} 
  (1-\delta)^{s-|c|}+\delta^{s-|c|}
  (1-\delta)^{|c|}\right)\pm 2n^{0.51}.
\end{equation*}

Now we prove that $\bU=(\bu^1,\ldots,\bu^s)$ is typical with high probability:

\begin{lemma}\label{lemma:easy1}
With probability at least
  $1-1/n^{\omega_n(1)}$, $\bU=(\bu^1,\ldots,\bu^s)$ is typical.
\end{lemma}

\begin{proof}
Let $\bz\sim \{0,1\}^n$, $\bU=(\bu^1,\ldots,\bu^s)$ be  independent
  draws from $\sphere_r(\bz)$. 
We show that 
\begin{equation}\label{eq:goal}
\Prx_{\bz,\bU}\big[\goodzU{\bz}{\bU}\big] \ge 1-n^{-\omega_n(1)},
\end{equation}
from which the lemma follows using Markov's inequality.

For the convenience of the analysis, consider the following  random process:
    \begin{flushleft}\begin{itemize}
        \item Draw $\bz \in \{0,1\}^n$ uniformly at random.
        \item For each $j \in [s]$, draw $\widehat{\bu}^j$ by flipping each bit of $\bz$ with probability $\delta$. Formally, for each index $i \in [n]$, independently set $\widehat{\bu}^j_i = \bz_i$ with probability $1 - \delta$ and $\widehat{\bu}^j_i = 1 - \bz_i$ with probability $\delta$.
        We write $\widehat{\bU}$ to denote the tuple $(\widehat{\bu}^1,\ldots,\widehat{\bu}^s)$.
    \end{itemize}\end{flushleft}
Since $r = \delta n$, by a standard bound {for the Binomial distribution},
we have
    \begin{equation*}
        \Prx_{\bz, \widehat{\bu}^j} \left[ \widehat{\bu}^j \in \sphere_r(\bz) \right] \ge  n^{-O(1)},
    \end{equation*}
 {for each $j \in [s]$}.
 Let $\widehat{\bU}=(\widehat{\bu}^1,\ldots,\widehat{\bu}^s)$. {Since $\widehat{\bu}^j$'s are independent of each other,}
     we further have
    \begin{equation*}
         \Prx_{\bz, \widehat{\bU}} \left[ \forall j \in [s], \widehat{\bu}^j \in \sphere_r(\bz) \right] \ge n^{-O(s)}.
    \end{equation*}
To connect with \Cref{eq:goal}, we note that 
  sampling $\bz$ and $\bU$ is the same as sampling $\bz$ and $\widehat{\bU}$
    conditioning on the event above.
As a result, it suffices to show that
\begin{equation}\label{eq:goal2}
\frac{\Prx_{\bz, \widehat{\bU}} \big[ \goodzU{\bz}{\widehat{\bU}}\ \text{violated} \big] }{\Prx_{\bz, \widehat{\bU}} \left[ \forall j \in [s], \widehat{\bu}^j \in \sphere_r(\bz) \right] }\le n^{-\omega_n(1)}.
\end{equation}

To bound the numerator, by Chernoff bound and a union bound over $c$, we have
    \begin{equation*}
        \Prx_{\bz, \widehat{\bU}} \left[ \goodzU{\bz}{\widehat{\bU}}\ \text{violated}  \right] \le 2^s\cdot 2^{-n^{\Omega(1)}}=2^{-n^{\Omega(1)}},
    \end{equation*}
using $\delta^s=1/n^{0.1}$ so that the expectation of $\col_{c}(\widehat{\bU})$ for any $c\in \{0,1\}^s$ is 
  at least $\Omega(\delta^s n)=\Omega(n^{0.9})$.
Then \Cref{eq:goal2} follows using the choice of $s\ll n^{\Omega(1)}$.   
\end{proof}

Finally we show that if $U$ is typical and $y$ is far from every $u^i$,
  then most likely $y\notin \sphere_r(\bz)$ over $\bz\sim \calZ_U$:


\begin{lemma}\label{lemma:easy2}
Let $U=(u^1,\ldots,u^s)$ be a typical tuple and let 
  $y$ be a point such that $\ham(y,u^i)\ge n^{0.4}$ for all $i\in [s]$.
Then 
\begin{equation} \label{eq:lovely}
\Prx_{\bz \sim \calZ_U}\big[\text{$y\in \sphere_r(\bz)$ and $\goodzU{\bz}{U}$}\big] \leq o(1/q).
\end{equation}
\end{lemma}

Before proving \Cref{lemma:easy2}, let's use it to prove 
  \Cref{lemma:maintechlower} first.

\begin{proof}[Proof of \Cref{lemma:maintechlower}]
Fix a $\Phi:(\{0,1\}^n)^s\rightarrow \{0,1\}^n$.
By \Cref{lemma:easy1}, $\bU$ is typical with probability at least $1-n^{-\omega_n(1)}$ (and note that $n^{-\omega_n(1)}=o(1/q)$).
Fix any $U$ that is typical, and let $y=\Phi(U)$.
If $\ham(y,u^i)< n^{0.4}$ for some $i$, we are done.
Otherwise, by \Cref{lemma:easy2}, we have 
  both $y\in \sphere_r(\bz)$ and $\goodzU{\bz}{U}$ with probability
  at most $o(1/q)$ over $\bz\sim \calZ_U$.
Given that  event $\goodzU{\bz}{U}$~holds with probability at least 
  $1-n^{-\omega_n(1)}$ using the definition of typical tuples,
  we have 
  $y\in \sphere_r(\bz)$ with probability at most $o(1/q)$ over $\bz\sim \calZ_U$.
This finishes the proof of \Cref{lemma:maintechlower}.
\end{proof}

\begin{proof}[Proof of \Cref{lemma:easy2}]
The proof analyzes two mutually exclusive cases.

\medskip

\noindent {\bf Case 1:} This is the case that there exists a string $c^*\in \{0,1\}^s$ such that 
$$
\big|i\in \col_{c^*}(U):y_i=1\big|\ge n^{0.2}\quad\text{and}\quad 
\big|i\in \col_{c^*}(U):y_i=0\big|\ge n^{0.2}.
$$
In this case, we reveal the following information about
  $\bz\sim \calZ_U$, denoted by $\calR$: 
$$
\text{$\calR$: the information is the bits
  $\bz_i$ for all $i\in [n]$ except those in $\col_{c^*}(U)$.}
$$

\begin{figure}[t!]
    \centering
    \makeatletter

\begin{tikzpicture}[yscale=-1]   
    
    \pgfmathsetlengthmacro{\WIDTH}{15cm}
    \pgfmathsetlengthmacro{\HEIGHT}{4cm}

    \pgfmathsetlengthmacro{\LMARGIN}{-1cm}
    \pgfmathsetlengthmacro{\RMARGIN}{\WIDTH+0.1cm}
    \pgfmathsetlengthmacro{\TMARGIN}{-2.5cm}
    \pgfmathsetlengthmacro{\BMARGIN}{\HEIGHT+2.2cm}

    \pgfmathsetlengthmacro{\COLUMNWIDTH}{1.8cm}

    \pgfmathsetlengthmacro{\COLUMN@C}{9cm}   
    \pgfmathsetlengthmacro{\COLUMN@C@MID}{\COLUMN@C+0.5*\COLUMNWIDTH}

    \pgfmathsetlengthmacro{\ROWmargin}{0.3cm}   
    \pgfmathsetlengthmacro{\ROWskip}{0.6cm}   
    \pgfmathsetlengthmacro{\ROWone}{\ROWmargin}   
    \pgfmathsetlengthmacro{\ROWtwo}{\ROWone+\ROWskip}   
    \pgfmathsetlengthmacro{\ROWs}{\HEIGHT-\ROWmargin}   
    \pgfmathsetlengthmacro{\ROWss}{\ROWs-\ROWskip}   
    \pgfmathsetlengthmacro{\ROWdots}{(\ROWtwo+\ROWss)/2}   

    \pgfmathsetlengthmacro{\ZBOTTOM}{-1cm}   
    \pgfmathsetlengthmacro{\ZTOP}{-1.6cm}   
    \pgfmathsetlengthmacro{\ROWz}{(\ZBOTTOM+\ZTOP)/2}   

    \pgfmathsetlengthmacro{\YSHIFT}{0.6cm}   
    \pgfmathsetlengthmacro{\YHEIGHT}{0.6cm}   
    \pgfmathsetlengthmacro{\YTOP}{\HEIGHT+\YSHIFT}   
    \pgfmathsetlengthmacro{\YBOTTOM}{\YTOP+\YHEIGHT}   
    \pgfmathsetlengthmacro{\ROWy}{(\YBOTTOM+\YTOP)/2}   
    \pgfmathsetlengthmacro{\YDECORATIONSHIFT}{0.8cm}   
    \pgfmathsetlengthmacro{\YDECORATION}{\YBOTTOM+\YDECORATIONSHIFT}
    \pgfmathsetlengthmacro{\YLINEOFFSET}{0.1cm}

    \pgfmathsetlengthmacro{\LEFTLABEL}{-0.7cm}   
    \pgfmathsetlengthmacro{\COLUMNLABEL}{-0.1cm}   

    \pgfmathsetlengthmacro{\BRACERAISE}{0.1cm}
    \pgfmathsetlengthmacro{\BRACEAMP}{5pt}   
    \pgfmathsetlengthmacro{\LABELSHIFT}{1em+\BRACEAMP-2pt}   
    \pgfmathsetlengthmacro{\SHORTBRACE}{1pt}   

    \pgfmathsetlengthmacro{\STTOPPOS}{\HEIGHT+\BRACERAISE}

    \pgfmathsetlengthmacro{\STOFFSET}{-3pt}
    \pgfmathsetlengthmacro{\STWIDE}{2pt}

    \pgfmathsetlengthmacro{\ZPOS}{\ZTOP-\BRACERAISE}

    \pgfmathsetlengthmacro{\YTOPPOS}{\YBOTTOM+\BRACERAISE}
    \pgfmathsetlengthmacro{\YLABELSHIFT}{1em+\BRACEAMP}   

    \colorlet{color@line@c}{black}
    \colorlet{color@line@cp}{black}
    \colorlet{color@text@c}{color@line@c}
    \colorlet{color@text@cp}{color@line@cp}
    \colorlet{color@label@c}{color@line@c}
    \colorlet{color@label@cp}{color@line@cp}
    \colorlet{color@decoration@c}{color@line@c}
    \colorlet{color@decoration@cp}{color@line@cp}
    \colorlet{color@emph}{violet}
    \colorlet{color@emph@c}{color@emph}
    \colorlet{color@emph@cp}{color@emph}
    \colorlet{color@emphbg}{white}
    \colorlet{color@randomness}{Yellow!40}

    \def\column@label@pre{\small}

    \draw[white] (\LMARGIN,\TMARGIN) rectangle (\RMARGIN,\BMARGIN);

    \draw (0,0) rectangle (\WIDTH,\HEIGHT);

    \draw[color@line@c] (\COLUMN@C,0) -- (\COLUMN@C,\HEIGHT);
    \draw[color@line@c] (\COLUMN@C+\COLUMNWIDTH,0) -- (\COLUMN@C+\COLUMNWIDTH,\HEIGHT);

    \node at (\LEFTLABEL,\ROWone) {$u^1$};
    \node at (\LEFTLABEL,\ROWtwo) {$u^2$};
    \node at (\LEFTLABEL,\ROWdots) {$\vdots$};
    \node at (\LEFTLABEL,\ROWss) {$u^{(s-1)}$};
    \node at (\LEFTLABEL,\ROWs) {$u^s$};

    \node[color@text@c] at (\COLUMN@C@MID,\ROWone) {$0~0~\compactcdots~0$};
    \node at (\COLUMN@C/2,\ROWone) {\textcolor{white}{0}$\cdots$\textcolor{white}{0}};   
    \node at (\COLUMN@C/2+\COLUMNWIDTH/2+\WIDTH/2, \ROWone) {\textcolor{white}{0}$\cdots$\textcolor{white}{0}};

    \node[color@text@c] at (\COLUMN@C@MID,\ROWtwo) {$1$};
    \node at (\COLUMN@C/2,\ROWtwo) {$\cdots$};
    \node at (\COLUMN@C/2+\COLUMNWIDTH/2+\WIDTH/2, \ROWtwo) {$\cdots$};

    \node[color@text@c] at (\COLUMN@C@MID,\ROWtwo/2+\ROWss/2) {$\vdots$};
    \node at (\COLUMN@C/2,\ROWtwo/2+\ROWss/2) {$\vdots$};
    \node at (\COLUMN@C/2+\COLUMNWIDTH/2+\WIDTH/2, \ROWtwo/2+\ROWss/2) {$\vdots$};

    \node[color@text@c] at (\COLUMN@C@MID,\ROWss) {$1$};
    \node at (\COLUMN@C/2,\ROWss) {$\cdots$};
    \node at (\COLUMN@C/2+\COLUMNWIDTH/2+\WIDTH/2, \ROWss) {$\cdots$};

    \node[color@text@c] at (\COLUMN@C@MID,\ROWs) {$0$};
    \node at (\COLUMN@C/2,\ROWs) {$\cdots$};
    \node at (\COLUMN@C/2+\COLUMNWIDTH/2+\WIDTH/2, \ROWs) {$\cdots$};

    \node[color@label@c, anchor=south] at (\COLUMN@C@MID,\COLUMNLABEL) (column-label-c) {\column@label@pre $\mathrm{col}_{c^*}(U)$};

    \fill[color@randomness] (\COLUMN@C,\ZTOP) rectangle (\COLUMN@C+\COLUMNWIDTH,\ZBOTTOM);
    
    \draw (0, \ZTOP) rectangle (\WIDTH, \ZBOTTOM);

    \node at (\LEFTLABEL,\ROWz) {$\boldsymbol{z}$};

    \draw[color@line@c] (\COLUMN@C,\ZTOP) -- (\COLUMN@C,\ZBOTTOM);
    \draw[color@line@c] (\COLUMN@C+\COLUMNWIDTH,\ZTOP) -- (\COLUMN@C+\COLUMNWIDTH,\ZBOTTOM);

    \node[color@text@c] at (\COLUMN@C@MID,\ROWz) {$0~1~\compactcdots~0$};
    \node at (\COLUMN@C/2,\ROWz) {\textcolor{white}{0}$\cdots$\textcolor{white}{0}};   
    \node at (\COLUMN@C/2+\COLUMNWIDTH/2+\WIDTH/2, \ROWz) {\textcolor{white}{0}$\cdots$\textcolor{white}{0}};

    

    \draw[color@decoration@c, decorate, decoration={brace, amplitude=\BRACEAMP}] (\COLUMN@C+\SHORTBRACE,\ZPOS) -- (\COLUMN@C+\COLUMNWIDTH-\SHORTBRACE,\ZPOS) node[color@decoration@c, midway, yshift=\LABELSHIFT] (z-c) {The only remaining randomness is $\boldsymbol{z}$ in $\mathrm{col}_{c^*}(U)$};

    \draw (0, \YTOP) rectangle (\WIDTH, \YBOTTOM);

    \node at (\LEFTLABEL,\ROWy) {$y$};
    
    \draw[color@line@c] (\COLUMN@C,\YTOP) -- (\COLUMN@C,\YBOTTOM);
    \draw[color@line@c] (\COLUMN@C+\COLUMNWIDTH,\YTOP) -- (\COLUMN@C+\COLUMNWIDTH,\YBOTTOM);

    \node[color@text@c] at (\COLUMN@C@MID,\ROWy) {$1~0~\compactcdots~0$};
    \node at (\COLUMN@C/2,\ROWy) {\textcolor{white}{0}$\cdots$\textcolor{white}{0}};   
    \node at (\COLUMN@C/2+\COLUMNWIDTH/2+\WIDTH/2, \ROWy) {\textcolor{white}{0}$\cdots$\textcolor{white}{0}};

    \draw[color@decoration@c, decorate, decoration={brace, mirror, amplitude=\BRACEAMP}] (\COLUMN@C+\SHORTBRACE,\YTOPPOS) -- (\COLUMN@C+\COLUMNWIDTH-\SHORTBRACE,\YTOPPOS) node[color@decoration@c, midway, yshift=-\YLABELSHIFT] {$y$ has at least $n^{0.2}$ $0$'s and $n^{0.2}$ $1$'s in $\mathrm{col}_{c^*}(U)$};
\end{tikzpicture}

\makeatother
    \caption{An illustration of Case~1 in the proof of \Cref{lemma:easy2}. The $2^s$ strings $w \in \{0,1\}^s$ induce a partition of $[n]$ into $2^s$ equivalence classes, where the $w$-th equivalence class is $\col_w(U)$. The column $\col_{c^*}(U)$ is the one of interest in this case. Note that for each index $i \in \col_w(U)$, the $i$-th column of $U$ is the same $s$-bit string, namely $w$; for example, for each $i \in \col_c(U)$, the $i$-th column of $U$ is the $s$-bit string $01 \dots 10$. For ease of visualization, the elements in each $\col_w(U)$ are depicted as forming a consecutive interval. Note that in this case, $y$ has at least $n^{0.2}$ $0$'s and $n^{0.2}$ $1$'s in $\col_{c^*}(U)$. After revealing $\calR$, the remaining randomness is $\bz$ in column $\col_{c^*}(U)$, marked in yellow in the figure.}
    \label{fig:lower_bound_easy2_case1}
\end{figure}

Given that $\bz\sim \calZ_U$, one can infer from $\calR$ the 
   number of $i\in \col_{c^*}(U)$ with $\bz_i=1$ and 
 the number of $i\in \col_{c^*}(U)$ with $\bz_i=0$.
 {To see this, let $\alpha=|\col_{c^*}(U)|$ (which is fixed by $U$) and 
 $\alpha_0$ (or $\alpha_1$) be the number of $\smash{i\in \col_{c^*}(U)}$ with $\bz_i=0$ (or $1$, respectively).
Taking $u^1$ from $U$ (or any $u$ in $U$~would work), $\bz$ must satisfy
$$
r=\ham(u^1,\bz)=\left|i\notin \col_{c^*}(U):u^1_i\ne \bz_i\right|
+\left|i\in \col_{c^*}(U):u^1_i\ne \bz_i\right|.
$$
The first term on the RHS is fixed given $\calR$ and
  thus, the second term is fixed as well.
Noting that $u^1_i=c_1^*$ for all $i\in \col_{c^*}(U)$, 
  the second term is $\alpha_{1-c_1^*}$; from
  $\alpha=\alpha_0+\alpha_1$ one can infer the other.}



On the one hand, as the event we care about {satisfies $\goodzU{\bz}{U}$},
  we may assume that
$$
\alpha_0=\frac{n}{2}\left(\delta^{|c^*|}(1-\delta)^{s-|c^*|}\right)\pm n^{0.51} \quad\text{and}\quad
\alpha_1=\frac{n}{2}\left(\delta^{s-|c^*|}(1-\delta)^{ |c^*|}\right)\pm n^{0.51}.
$$
Plugging in $\delta^s=1/n^{0.1}$, we have that both $\alpha_0$
  and $\alpha_1$ are at least $n^{0.9}/3$ and thus, both $\alpha_0/\alpha$ and $\alpha_1/\alpha$ are $\Omega(1/n^{0.1})$.
On the other hand, with $\calR$ fixed, the only randomness left in $\bz$ (conditioning on $U$ and $\calR$)
  is to uniformly set $\alpha_0$ many $\bz_i$ to be $0$ from $i\in \col_{c^*}(U)$, and the rest to be $1$.
Let $\bk$ be the random variable that denotes the number of $i\in \col_{c^*}(U)$ with $y_i=0$ and $\bz_i=1$.
Using the assumption that $y$ has at least $n^{0.2}$ many $0$'s and at least $n^{0.2}$ many $1$'s, we have that  the probability of $\bk=k$ for any fixed value of $k$ is $O(1/n^{ {0.05}})$ by the following claim:

\begin{claim}\label{claim:simple1}
For any fixed $k$, the probability of $\bk=k$ is at most $O(1/n^{ {0.05}})$.
\end{claim}
\begin{proof}
Let $m_0$ (or $m_1$) be the number of $i\in \col_{c^*}(U)$ 
  with $y_i=0$ (or $1$, respectively).
Recall that 
$$n^{0.2}\le m_0,m_1\le n,\quad
  n^{0.9}/3\le \alpha_0,\alpha_1\le n\quad \text{and}\quad  \alpha=\alpha_0+\alpha_1=m_0+m_1.$$
Then the total number of possible $\bz$ with $\bk=k$ is 
\begin{equation}\label{eq:hehe11}
{m_0\choose k}{m_1\choose \alpha_1-k}.
\end{equation}
We consider two cases: $k\le \beta:=\alpha_1(m_0/\alpha)$ and $k>\beta$.
When $k\le \beta$, we consider 
$$
N_\Delta:={m_0\choose k+\Delta}{m_1\choose \alpha_1-k-\Delta}
$$
and show that $N_0$ (which is just \Cref{eq:hehe11}) is $\Theta(\cdot)$
  of $N_1,\ldots,N_{\Omega(n^{0.05})}$. The lemma in this case then follows.
To this end we first note that 
$$
m_0-k\ge m_0-\beta=m_0-\alpha_1\cdot \frac{m_0}{\alpha}
=m_0\cdot \frac{\alpha_0}{\alpha}\ge n^{0.2}\cdot \frac{n^{0.9}/3}{n}\ge \Omega(n^{0.1}),
$$
and similarly $\alpha_1-k\ge \alpha_1-\beta=\alpha_1m_1/\alpha\ge  \Omega(n^{0.1})$. 
So all binomials involved in $N_1,\ldots,N_{\Omega(n^{0.05})}$ are well defined.
Examining $N_{\Delta+1}/N_\Delta$ for any $\Delta=0,1,\ldots,\Omega(n^{0.05})$, we have 
\begin{align*}
\frac{N_{\Delta+1}}{N_\Delta} =\frac{m_0-k-\Delta}{k+\Delta+1}\cdot
\frac{\alpha_1-k-\Delta}{m_1-\alpha_1+k+\Delta+1} \ge  \frac{m_0-\beta-\Delta}{\beta+\Delta+1}\cdot \frac{\alpha_1-\beta-\Delta}{m_1-\alpha_1+\beta+\Delta+1}.
\end{align*}
Given that 
$$m_0-\beta=\frac{m_0\alpha_0}{\alpha},\quad 
\alpha_1-\beta=\frac{\alpha_1 m_1}{\alpha},\quad
\beta=\frac{\alpha_1m_0}{\alpha}\quad \text{and}\quad 
m_1-\alpha_1+\beta=\frac{m_1\alpha_0}{\alpha},$$
they are all at least $\Omega(n^{0.1})$ and thus,  
$$
\frac{N_{\Delta+1}}{N_{\Delta}}\ge \frac{(1\pm O(1/n^{0.05})) \cdot m_0\alpha_0/\alpha}{(1\pm O(1/n^{0.05}))\cdot \alpha_1m_0/\alpha}\cdot 
\frac{(1\pm O(1/n^{0.05}))\cdot \alpha_1m_1/\alpha}{(1\pm O(1/n^{0.05}))\cdot m_1\alpha_0/\alpha}=1\pm O\left(\frac{1}{n^{0.05}}\right).
$$
From this we have $N_\Delta=\Omega(N_0)$ for $\Delta$ up to some $\Omega(n^{0.05})$.

The other case when $k>\beta$ is symmetric, where we examine 
  $N_0,N_{-1},\ldots,N_{-\Omega(n^{0.05})}$.
\end{proof}
  
Given that $\bz_i$'s are all fixed in $\calR$ outside of $\col_{c^*}(U)$, we have that
  $y\in \sphere_r(\bz)$ if and only if the number of $i\in \col_{c^*}(U)$
  with $y_i\ne \bz_i$ is exactly
$$
r-\left|i\notin \col_{c^*}(U): y_i\ne \bz_i\right|.
$$
On the other hand, this number is also uniquely determined by $\bk$, since it is 
$$
\bk+\big|i\in \col_{c^*}(U):y_i=1\big|-(\alpha_1-\bk)
=2\bk+\big|i\in \col_{c^*}(U):y_i=1\big|-\alpha_1.
$$
So there is a unique value of $\bk$ that can put 
  $y\in \sphere_r(\bz)$.
Using \Cref{claim:simple1}, the conclusion of \Cref{lemma:easy2} holds in Case~1.



\medskip

\noindent {\bf Case 2:} In this case, for every $c\in \{0,1\}^s$,  $y$ satisfies either
  $$
\big|i\in \col_{c }(U):y_i=1\big|<n^{0.2}\quad\text{or}\quad 
\big|i\in \col_{c }(U):y_i=0\big|< n^{0.2}.
  $$
As a result, we can round $y$ to obtain a  string $y^*\in \{0,1\}^n$ such that 
\begin{flushleft}\begin{enumerate}
\item For every $c\in \{0,1\}^s$, $y^*_i$ is the same for all $i\in \col_c(U)$. So $y^*$ can be succinctly represented by a $2^s$-bit string $b$ indexed by $c\in \{0,1\}^s$:
  $y^*_i=b_c$ for all $i\in \col_c(U)$; and
\item $\ham(y,y^*)\le 2^s \cdot n^{0.2}< n^{0.21}$.
\end{enumerate}\end{flushleft}
Notice that the latter property implies 
 that $y^*$ is different from all $u^i$ in $U$, given that $\ham(y,u^i)\ge n^{0.4}$ for all $u^i$.
It also implies that a necessary condition for
  $y\in \sphere_r(\bz)$ is
  $|\ham(y^*,\bz)-r|\le n^{0.21}$.
We show below that the probability of
\begin{equation}\label{eq:theevent}
\text{both $\big|\ham(y^*,\bz)-r\big|\le n^{0.21}$ and  {$\goodzU{\bz}{U}$}}
\end{equation}
is at most $o(1/q)$ when $\bz\sim \calZ_U$, which gives \Cref{lemma:easy2} in Case~2.

To this end we first observe that $b_{0^s}$ must be $0$
  and $b_{1^s}$ must be $1$, since if otherwise, say $b_{0^s}=1$,
  then whenever $\bz$ satisfies $\goodzU{\bz}{U}$, we have
$$
\left|i\in \col_{0^s}(U):\bz_i=0\right|\ge 
\frac{n}{2}(1-\delta)^s -n^{0.51}\ge \frac{n}{3}
$$
and thus, $\ham(y^*,\bz)\ge n/3\gg r$.
On the other hand, using $b_{0^s}=0$ and $b_{1^s}=1$, there 
  must exist 
an $i\in [s]$ and a string $c\in \{0,1\}^s$ with $c_i=0$ such that 
  $b_{c}=0$ and $b_{c^{(i)}}=1$, {where $c^{(i)}$ denotes $c$ with the $i$-th bit flipped from 0 to 1}.
To see this, consider {any} path of length $s$ from $0^s$ to $1^s$ along 
  which we flip {some} bit from $0$ to $1$ in each step.
Given that $b_{0^s}=0$ and $b_{1^s}=1$ such an $i\in [s]$ must exist for some string $c$ occurring on the path.
Without loss of generality, we can also assume that $i=s$ after renaming the rows 
so there exists a $c\in \{0,1\}^s$ with $c_s=0$
  such that $b_c=0$ and $b_{c^{(s)}}=1$.
Furthermore, using $y^*\ne u^s$, there must 
  exist a string $c'\in \{0,1\}^s$ with $c'_s=0$ such that  $$(0,1)\ne \left(b_{c'},b_{c'^{(s)}}\right)\in \left\{(1,0),(0,0),(1,1)\right\}$$
 {(because if for every string $c' \in \zo^s$ with $c'_s=0$ we had $(0,1)=(b_{c'},b_{c'^{s)}})$, then $y^*$ would be precisely $u^{s}$). We remark that since $(b_c,b_{c^{(s)}})=(0,1) \neq (b_c',b_{c'^{(s)}})$, we clearly have $c \neq c'$.}

\begin{figure}[t!]
    \centering
    \includestandalone{assets/lower_bound_query_figure_case_2}
    \caption{An illustration of Case~2 in the proof of \Cref{lemma:easy2}.  The figure follows the same structure as \Cref{fig:lower_bound_easy2_case1}, which illustrates four different strings $w$: $c$, $c^{(s)}$, $c'$, and $c'^{(s)}$. In this case, $y$ is either all-$0$ or all-$1$ in any of the columns~(not only in the four columns in the figure). After revealing $\calR$, the remaining randomness is $\bz$ in $\col_c(U),\col_{c^{(s)}}(U)$, $\col_{c'}(U)$ and $\col_{c'^{(s)}}(U)$ (marked in yellow), \emph{subject to} known $\alpha_0$, $\alpha_1$, $\beta_0$, $\beta_1$, and $\gamma$.}
    \label{fig:lower_bound_easy2_case2}
\end{figure}


Consider the case when $(b_{c'},b_{c'^{(s)}})=(1,1)$ (the other
  two cases will be handled similarly). We will show that the event described 
  in \Cref{eq:theevent} happens with probability $o(1/q)$ when 
  $\bz\sim \calZ_{U}$.
For convenience, we write $S_0$ to denote $\col_c(U)$, $S_1$ to 
  denote $\smash{\col_{c^{(s)}}(U)}$, $T_0$ to denote $\col_{c'}(U)$,
  $T_1$ to denote $\col_{c'^{(s)}}(U)$, $S=S_0\cup S_1$ and $T=T_0\cup T_1$.
{(Note that since $c \neq c'$, the four sets $S_0,S_1,T_0,T_1$ are disjoint subsets of $[n]$.)}
Since $U$ is typical, we have
\begin{align*}
|S_0|&=\frac{n}{2}\left(\delta^{|c|}(1-\delta)^{s-|c|}
+\delta^{s-|c|}(1-\delta)^{|c|}\right)\pm 2n^{0.51}\\[0.6ex]
|S_1|&=\frac{n}{2}\left(\delta^{|c|+1}(1-\delta)^{s-|c|-1}
+\delta^{s-|c|-1}(1-\delta)^{|c|+1}\right)\pm 2n^{0.51}
\end{align*}
and similar equations hold for $|T_0|$ and $|T_1|$ as well (with $|c|$ replaced by $|c'|$).

{Before moving on, we would like to reminder the reader that $U$ is fixed and thus,
  sets $\col_c(U)$ are fixed for all $c$ and in particular, $S_0,S_1,S$ as well as $T_0,T_1,T$
  are all fixed; what is random here is the string $\bz\sim \calZ_U$ (see \Cref{fig:lower_bound_easy2_case2}).}

In the analysis of $\bz\sim \calZ_U$, we first reveal the following information about $\bz$,
  denoted by $\calR$:
\begin{flushleft}\begin{enumerate}
\item $\bz_i$ for all $i\in [n]$ outside of $S\cup T$; and 
\item the number of $0$'s in $\bz$ in $S$, which also 
  reveals the number of $1$'s in $\bz$ in $S$ {(this is because, as mentioned
  earlier, $S$ is fixed; so knowing the number of $0$'s in $S$
  reveals the number of $1$'s).}
\end{enumerate}\end{flushleft}
Given that $u^1 \in \sphere_r(\bz)$ (or any string in $U$ other than $u^s$ can be used for this argument), 
  one can directly infer from $\calR$ the following parameter:
\begin{flushleft}\begin{enumerate}
\item[3.] the number of $0$'s in $\bz$ in $T$ as well as the number of $1$'s in $\bz$ in $T$. {To see this, we have
$$
r=\ham(u^1,\bz)=\left|\big\{i\notin S\cup T:\bz_i\ne u^1_i\big\}\right|
+\left|\big\{i\in S :\bz_i\ne u^1_i\big\}\right|+\left|\big\{i\in T:\bz_i\ne u^1_i\big\}\right|.
$$
The first number on the RHS is fixed given $\calR$.
The second number is also fixed given that $u_i^1=c_1$ for all $i\in S$ and thus,
  it is just the number of $i\in S$ with $\bz_i=1-c_1$.
As a result, the last number on the RHS can be inferred.
Given that $u_i^1=c'_1$ for all $i\in T$,
  this is exactly the number of $i\in T$ with $\bz_i=1-c'_1$.}
\end{enumerate}
\end{flushleft}
Given that $u^s\in \sphere_r(\bz)$, one can also infer 
\begin{flushleft}
\begin{enumerate}
\item[4.]
  the number of $i\in S\cup T$ with $u^s_i\ne \bz_i$. {To see this, similarly we have
$$
r=\ham(u^s,\bz)=\left|\big\{i\notin S\cup T:\bz_i\ne u^s_i\big\}\right|
+ \left|\big\{i\in S\cup T:\bz_i\ne u^s_i\big\}\right|.
$$
The first number on the RHS is fixed given $\calR$ so the second number can be inferred.}
\end{enumerate}\end{flushleft}
For convenience, let $\alpha_{ 0}$
  (or $\alpha_1$) be the number of $0$'s (or $1$'s) in $\bz_i$ for $i\in S$, and 
similarly let
  $\beta_{0}$ (or $\beta_1$) be the number of $0$'s (or $1$'s) in $\bz_i$ for $i\in T$.
So $|S|=\alpha_0+\alpha_1$ and $|T|=\beta_0+\beta_1$.
Also we let $\gamma$ be the number of $i\in S\cup T$ with
  $u_i^s\ne \bz_i$.

\def\bell{\boldsymbol{\ell}}
  
The rest of the randomness of $\bz$ (conditioning on both $U$ and $\calR$)
  is as follows:
$\bz$ over $S\cup T$ is a string uniformly drawn
  from all strings that satisfy the following three conditions
  (see \Cref{fig:lower_bound_easy2_case2}):
\begin{enumerate}
\item There are $\alpha_0$ many $0$'s and 
  $\alpha_1$ many $1$'s in $S$;
\item There are $\beta_0$ many $0$'s and $\beta_1$
  many $1$'s in $T$; and
\item The number of $i\in S\cup T$ with $u^s_i\ne \bz_i$ is $\gamma$.
\end{enumerate}
{This is because, as argued earlier, every $\bz\sim \calZ_U$ that fits $\calR$
  must satisfy these conditions; and~when these conditions are met,
  $\bz$ is a string in the support of $\calZ_U$ that fits $\calR$.}

Four random variables left undetermined in $\bz$ are (1) the number of $i\in S_0$ with $\bz_i=1$, which we denote $\bk_0$; (2) the number of $i\in S_1$ with $\bz_i=1$, which we denote $\bk_1$; (3) the number of $i\in T_0$ with $\bz_i=1$, which we denote $\bell_0$; and (4) the number of $i\in T_1$ with $\bz_i=1$, which we denote $\bell_1$.
These are not completely free variables since they must satisfy the following equations:
\begin{align*}
\bk_0+\bk_1&=\alpha_1\\
\bell_0+\bell_1&=\beta_1\\
\bk_0+(|S_1|-\bk_1)+\bell_0+(|T_1|-\bell_1)&=\gamma.
\end{align*}
Solving the system of linear equations, we get that $\bk_1,\bell_0$ and $\bell_1$
  are all determined by $\bk_0$:
\begin{equation}\label{eq:system}
\bk_1=\alpha_1-\bk_0,\quad
\bell_0=\frac{\gamma+\alpha_1+\beta_1-|S_1|-|T_1|}{2}-\bk_0,\quad\text{and}\quad
\bell_1=\frac{\beta_1+|S_1|+|T_1|-\gamma-\alpha_1}{2}+\bk_0.
\end{equation}
We write $Z_k$ to denote the number of $\bz$ (over $S\cup T$) with $\bk_0=k$.

We note that 
\begin{align*}
\ham(y^*,\bz)&=\left|i\notin S\cup T:y^*_i\ne \bz_i\right|+\beta_0+\bk_0+(|S_1|-\bk_1)\\[0.8ex]
&=\left|i\notin S\cup T:y^*_i\ne \bz_i\right|+\beta_0+
|S_1|-\alpha_1+2\bk_0.
\end{align*}
Thus for $y^*$ to satisfy $|\ham(y^*,\bz)-r|\le n^{0.21}$, we need to have $|\bk_0-\kappa_0|\le 2n^{0.21}$, where $\kappa_0$ is a fixed number given by
$$
\kappa_0=\left\lceil \frac{r-|i\notin S\cup T:y_i^*\ne \bz_i|-\beta_0-|S_1|+\alpha_1}{2}\right\rceil.
$$
Let $\kappa_1,\lambda_0,\lambda_1$ be the values of $\bk_1,\bell_0,\bell_1$
  by plugging $\bk_0=\kappa_0$ in \Cref{eq:system}. 
We may further assume without loss of generality that $\kappa_0,\kappa_1,\lambda_0,\lambda_1$
  satisfy
\begin{align*}
\kappa_0= \frac{n}{2}\left(\delta^{s-|c|}(1-\delta)^{|c|}\right)\pm 2n^{0.51} \quad&\text{and}\quad
\kappa_1= \frac{n}{2}\left(\delta^{s-|c|-1}(1-\delta)^{|c|+1}\right)\pm 2n^{0.51}\\[1ex]
\lambda_0= \frac{n}{2}\left(\delta^{s-|c'|}(1-\delta)^{|c'|}\right)\pm 2n^{0.51}\quad&\text{and}\quad
\lambda_1= \frac{n}{2}\left(\delta^{s-|c'|-1}(1-\delta)^{|c'|+1}\right)\pm 2n^{0.51}
\end{align*}
since otherwise, $|\bk_0-\kappa_0|\le 2n^{0.21}$ and $\bk_0,\bk_1,\bell_0,\bell_1$  
 {such that $\bz$ satisfies $\goodzU{\bz}{U}$}
  cannot happen at the same time. 
  {For example, if $\kappa_0$ violates the condition above, then 
  $|\bk_0-\kappa_0|\le 2n^{0.21}$ implies 
$$
\left|\bk_0-\frac{n}{2}\left(\delta^{s-|c|}(1-\delta)^{|c|}\right)\right|
\ge |\bk_0-\kappa_0|+\left|\kappa_0-\frac{n}{2}\left(\delta^{s-|c|}(1-\delta)^{|c|}\right)\right|\ge 2n^{0.51}-2n^{0.21}>n^{0.51},
$$    
which violates $\goodzU{\bz}{U}$ (as in \Cref{def:good}); on the other hand, we have 
  $|\bk_0-\kappa_0|= |\bk_1-\kappa_1|=|\bell_0-\lambda_0|=|\bell_1-\lambda_1|$
  by \Cref{eq:system} so the same argument works for $\kappa_1,\lambda_0$ and $\lambda_1$ as well.
  }
  
It suffices to show that 
$$
\frac{\sum_{k:|k-\kappa_0|\le 2n^{0.21}} Z_k}{\sum_k Z_k}=o\left(\frac{1}{q}\right).
$$
To this end, we have 
$$\displaystyle
\frac{\sum_{k:|k-\kappa_0|\le 2n^{0.21}} Z_k}{\sum_k Z_k}\le 
\frac{\sum_{k:|k-\kappa_0|\le 2n^{0.21}} Z_k}{\sum_{k:|k-\kappa_0|\le n^{0.4}} Z_k}=
\frac{\sum_{\Delta=-2n^{0.21}}^{2n^{0.21}}{|S_0|\choose \kappa_0+\Delta}{|S_1|\choose \kappa_1-\Delta }{|T_0|\choose \lambda_0-\Delta}{|T_1|\choose \lambda_1+\Delta}}
{\sum_{\Delta=-n^{0.4}}^{n^{0.4}}
{|S_0|\choose \kappa_0+\Delta}{|S_1|\choose \kappa_1-\Delta }{|T_0|\choose \lambda_0-\Delta}{|T_1|\choose \lambda_1+\Delta}}
$$
and it suffices to show that all $Z_{\kappa_0+\Delta}$, when $|\Delta|\le n^{0.4}$,
  are {multiplicatively $(1 \pm o(1))$-close} to $Z_{\kappa_0}$ (when this is true, the RHS can be upper bounded by $O(1/n^{0.19})$, which is $o(1/q)$ given $q=n^{0.01}$).

To see all $Z_{\kappa_0+\Delta}$ are close to $Z_{\kappa_0}$, we take any
  $\Delta$ with $|\Delta|\le n^{0.4}$ and have
\begin{align*}
\frac{Z_{\kappa_0+\Delta+1}}{Z_{\kappa_0+\Delta}}
&=\frac{|S_0|-\kappa_0-\Delta}{\kappa_0+\Delta+1}\cdot 
\frac{\kappa_1-\Delta}{|S_1|-\kappa_1+\Delta+1}
\cdot \frac{\lambda_0-\Delta}{|T_0|-\lambda_0+\Delta+1}\cdot \frac{|T_1|-\lambda_1-\Delta}{\lambda_1+\Delta+1}
\\[0.8ex]
&=  \left(\frac{1-\delta}{\delta}\right)^{s-2|c|}
\cdot \left(\frac{\delta}{1-\delta}\right)^{s-2|c|-2}\cdot \left(\frac{\delta}{1-\delta}\right)^{s-2|c'|}\cdot 
\left(\frac{1-\delta}{\delta}\right)^{s-2|c'|-2}\cdot 
\left(1\pm O\left(\frac{\delta}{n^{0.9}} \right)\right)\\[0.6ex]
&= 1\pm O\left(\frac{1}{\sqrt{n}}\right),
\end{align*}
using our previous estimates about $\kappa_0,\kappa_1,\lambda_0,\lambda_1$ and $|S_0|,|S_1|,|T_0|,|T_1|$.
It follows that $$Z_{\kappa_0+\Delta}=(1\pm o_n(1))\cdot Z_{\kappa_0}$$ for all $\Delta$ with $|\Delta|\le n^{0.4}$ and the claim follows.

The cases of $(0,0)$ and $(1,0)$ are similar. 
The only difference is that for $(0,0)$, we have 
\begin{align*}
\ham(y^*,\bz)&=\big|i\notin S\cup T:y_i^*\ne \bz_i\big|
+\beta_1+\bk_0+(|S_1|-\bk_1) =\text{some fixed integer}+2\bk_0 
\end{align*}
and for $(1,0)$, we have
\begin{align*}
\ham(y^*,\bz)&=\big|i\notin S\cup T:y_i^*\ne \bz_i\big|
+ \bk_0+(|S_1|-\bk_1)+(|T_0|-\bell_0)+\bell_1\\[0.6ex]
&=\text{some fixed integer}+4\bk_0. 
\end{align*}
In both cases, $\ham(y^*,\bz)$ is uniquely determined by $\bk_0$ and 
  the rest of the proof is the same.
This finishes the proof of \Cref{lemma:easy2}.
\end{proof}

\section{Discussion and Open Problems} \label{sec:discussion}

\noindent {\bf (Near-)Optimality of the analysis of our lower bound construction.} 
Our main theorem shows that no algorithm which is given $\frac{0.05\log n}{\log \log n}$ samples from $\SAMP(\boldf)$ and can make $n^{0.01}$ black-box oracle calls to $\boldf$ can reliably determine whether $\boldf \sim \Dyes$ or $\boldf \sim \Dno$.  
We remark here that this result is close to optimal for the $\Dyes$ and $\Dno$ distributions that we consider; more precisely, an algorithm that is given $100 \log n$ samples from $\SAMP(\boldf)$ and can make 100 black-box oracle calls to $\boldf$ can determine, with high constant accuracy, whether $\boldf \sim \Dyes$  or $\boldf \sim \Dno$. 
At a high level, this is because such an algorithm can use the $100 \log n$ samples to exactly identify the ``center'' $\bz$ of the unknown yes- or no- function, and then with knowledge of $\bz$ it can use 100 queries to determine whether the function is a yes- function or a no- function.  

In a bit more detail, this is done as follows:
\begin{flushleft}\begin{itemize}

\item The arguments of \Cref{sec:sample-based} that are used to prove \Cref{claim:samp-distribution-similar} show that with $1-o(1)$ probability, the distribution of $100 \log n$ samples drawn from either $\boldf \sim \Dyes$ or $\boldf \sim \Dno$ will be identical to the distribution of $100 \log n$ samples drawn from $\sphere_r(\bz)$ where $\bz$ is uniform random over $\zo^n$.

\item For any coordinate $i \in [n]$, a simple probabilistic argument using the Chernoff bound shows that with probability $1 - o(1/n)$, the majority vote of the $i$-th coordinate of $100 \log n$ samples drawn from $\sphere_r(\bz)$ will equal $\bz_i$.  Given this, a union bound over all $n$ coordinates shows that with probability $1-o(1)$, an algorithm that receives $100 \log n$ samples drawn from either $\boldf \sim \Dyes$ or $\boldf \sim \Dno$ can completely determine the string $\bz$.

\item Given the string $\bz$, the algorithm queries $100$ uniform random points at Hamming distance exactly $r$ from $\bz$. In the yes- case, the black-box oracle will respond 1 to all of these queries.  In the no- case, (a simplified version of) the arguments used to prove \Cref{claim:no-LTF} show that with probability at least $999/1000$ over the draw of the random bits $\ba_{\bar{b}}$ in Step~3 of the description of the no- distribution from \Cref{sec:yes}, the black-box oracle will respond 0 to at least one of the 100 queries.
\end{itemize}\end{flushleft}


\medskip

\noindent {\bf A $\poly(n/\eps)$ upper bound.}
In \cite{DDS15}, De et al.~gave an algorithm that uses $\poly(n/\eps)$ independent uniform samples from $f^{-1}(1)$ to \emph{learn} any unknown halfspace $f$ over $\zo^n$ to $\eps$-accuracy in relative error (see Theorem~1.2 of \cite{DDS15}).  
As is well known (see Proposition~3.1 of \cite{GGR98}), uniform-distribution PAC learning for a class of functions implies standard-model property testing with essentially the same complexity, and it is not difficult to establish this implication in the relative error setting as well.  More precisely, the relative-error  learning algorithm of \cite{DDS15} yields a testing algorithm in our relative-error model, by simply running the learning algorithm to obtain a hypothesis halfspace $h$ and then drawing $O(1/\eps)$ random samples from $f^{-1}(1)$ (respectively $h^{-1}(1)$) and checking that they are also satisfying assignments of $h$ (respectively $f$).  Thus, the \cite{DDS15} learning result yields a relative-error $\eps$-testing algorithm for halfspaces that makes $\poly(n/\eps)$ draws from $\SAMP(f)$ and $O(1/\eps)$ calls to $\MQ(f)$.

\medskip

\noindent {\bf Directions for future work.}  A natural goal for future work is to give a more precise characterization of the difficulty of testing halfspaces with relative error.  As sketched above, the analysis of our lower bound construction is essentially best possible, but is it possible to give an alternate construction which would lead to an improved lower bound?

Another appealing goal is to try to develop a non-trivial relative-error testing algorithm for halfspaces. While our main result shows that $o({\frac {\log n}{\log \log n}})$ complexity is unachievable, perhaps it is possible to do better than the naive approach based on relative-error learning that is sketched above.  Specifically, does there exist a relative-error halfspace tester with \emph{sublinear} (i.e.~$\poly(1/\eps) \cdot o(n)$) sample and query complexity?  Some initial progress towards this goal has been achieved in \cite{anonymous-halfspace}, which gives various algorithms for relative-error testing of halfspaces under the standard $N(0,I_n)$ Gaussian distribution which, under mild conditions, have complexity $\poly(1/\eps) \cdot o(n)$.

\section*{Acknowledgements}

X.C.~is supported by NSF grants CCF-2106429 and CCF-2107187. 
A.D.~is supported by NSF grant CCF 2045128. 
Y.H.~is supported by NSF grants CCF-2211238, CCF-2106429, and CCF-2238221
R.A.S.~is supported by NSF grants CCF-2211238 and CCF-2106429. 
T.Y.~is supported by NSF grants CCF-2211238, CCF-2106429, and AF-Medium 2212136.
T.Y.~and Y.H.~are also supported by an Amazon Research Award, Google CyberNYC award, and NSF grant CCF-2312242. 

\begin{flushleft}
\bibliographystyle{alphaabbr}
\bibliography{allrefs}
\end{flushleft}

\end{document}